%% file: main.tex
\documentclass[%
 reprint,
superscriptaddress,
 amsmath,amssymb,
 aps,
prx,
floatfix,
longbibliography
]{revtex4-2}

\usepackage{ifpdf}
\usepackage{braket}
\usepackage{amsmath} 
\usepackage{amssymb}
\usepackage{amsfonts}
\usepackage{amsthm}
\usepackage{braket}
\usepackage{xcolor}
\usepackage[colorlinks=true,linkcolor=blue,citecolor=blue,breaklinks]{hyperref}
\usepackage{graphicx}
\usepackage{dcolumn}
\usepackage{bm}
\usepackage[normalem]{ulem}
\usepackage{verbatim}
\usepackage{multibib}

\newcites{SM}{SM References}

\newtheorem{theorem}{Theorem}
\newtheorem{lemma}{Lemma}
\newtheorem{definition}{Definition}

\newtheorem{corollary}{Corollary}

\newcommand{\la}{\langle}
\newcommand{\ra}{\rangle}
\newcommand{\tr}{\mathrm{Tr}}

\newcommand{\Mmin}{M_{\mathrm{min}}}
\newcommand{\Mmax}{M_{\mathrm{max}}}
\newcommand{\id}{\mathrm{id}}

\definecolor{Zcolour}{rgb}{0.992, 0.588, 0.22}
\definecolor{purple}{rgb}{0.5,0,0.5}
\definecolor{brown}{rgb}{0.6,0.2,0}
\definecolor{dkgreen}{rgb}{0,0.5,0}

\usepackage{mathtools}

\newtheorem{innercustomthm}{Proposition}
\newenvironment{customthm}[1]{\renewcommand\theinnercustomthm{#1}\innercustomthm}{\endinnercustomthm}

\begin{document}


\title{Covariant Quantum Error-Correcting Codes\\ with Metrological Entanglement Advantage}
\author{Cheng-Ju Lin}
\thanks{cjlin@umd.edu}
\affiliation{Joint Center for Quantum Information and Computer Science, NIST/University of Maryland, College Park, Maryland 20742, USA}
\affiliation{Joint Quantum Institute, NIST/University of Maryland, College Park, Maryland 20742, USA}

\author{Zi-Wen Liu}
\thanks{zwliu0@tsinghua.edu.cn}
\affiliation{Yau Mathematical Sciences Center, Tsinghua University, Beijing 100084, China}
\affiliation{Perimeter Institute for Theoretical Physics, Waterloo, Ontario N2L 2Y5, Canada}

\author{Victor V. Albert}
\thanks{vva@umd.edu}
\affiliation{Joint Center for Quantum Information and Computer Science, NIST/University of Maryland, College Park, Maryland 20742, USA}

\author{Alexey V. Gorshkov}
\thanks{gorshkov@umd.edu}
\affiliation{Joint Center for Quantum Information and Computer Science, NIST/University of Maryland, College Park, Maryland 20742, USA}
\affiliation{Joint Quantum Institute, NIST/University of Maryland, College Park, Maryland 20742, USA}

\date{\today}

\begin{abstract}
We show that a subset of the basis for the irreducible representations of a tensor-product SU(2) rotation forms a covariant approximate quantum error-correcting code with transversal U(1) logical gates.
Generalizing previous work on ``thermodynamic codes" to general local spin and different irreducible representations using only properties of the angular momentum algebra, we obtain bounds on the code inaccuracy under generic noise on any known $d$ sites, under independent and identically distributed noise, and under heralded $d$-local erasures. 
We demonstrate that this family of codes protects a probe state with quantum Fisher information surpassing the standard quantum limit when the sensing parameter couples to the generator of the U(1) logical gate.
\end{abstract}

\maketitle

Quantum error-correcting codes have found applications beyond fault-tolerant quantum computation to a wide array of quantum technologies and physical scenarios, each requiring unique features.
For quantum metrology, protected encodings allow for a metrological \textit{entanglement advantage}, where the ability to measure a signal, quantified by the quantum Fisher information (QFI), surpasses the standard quantum limit~\cite{kesslerQuantum2014,arradIncreasing2014,durImproved2014,luRobust2015a,matsuzakiMagneticfield2017,sekatskiQuantum2017,zhouAchieving2018,zhouOptimal2020,ouyangRobust2022,faistTimeenergy2023,rojkovBias2022,ouyangFiniteround2024}.
Another example is code covariance with continuous transversal (i.e., tensor-product) gates, which is relevant in scenarios like reference-frame error correction~\cite{bartlettReference2007,haydenError2021}.
While these gates are incompatible with exact finite-dimensional codes due to the Eastin-Knill theorem~\cite{eastinRestrictions2009}, they are achievable with approximate quantum error-correcting codes (AQECCs)~\cite{leungApproximate1997,crepeau2005Approximate,benyGeneral2010,yiComplexity2024}.

One can then consider the task of measuring the variable parametrizing a continuous transversal gate for metrological purposes.
Various studies have established bounds on the inaccuracy of covariant AQECCs from this perspective~\cite{faistContinuous2020,zhouNew2021,kubicaUsing2021,PRXQuantum.3.010337,liuQuantum2022,liuApproximate2023}.
Conversely, codes with continuous gates can be candidates for QFI protection, though code covariance is not necessary for this application.

Motivated by these ideas, we introduce a family of AQECCs with transversal continuous gates, based on simple symmetry considerations. Additionally, these codes provide a metrological entanglement advantage for certain code parameters.
More specifically, we consider $N$ spin-$s$ degrees of freedom and the irreducible representations (irreps) of their total SU(2) rotation. 
The irrep basis is $|J,\!M\ra$, where $J \leq sN$ is the total spin quantum number, and $M$ is the magnetic quantum number.
We show that the states $|J,\!M\ra$, with certain choices of $M$, form an AQECC with transversal U(1) logical gates, and provide bounds on its code inaccuracy under generic noise on any known $d$ sites, under independently and identically distributed (i.i.d.) noise for $J \leq sN$, and under heralded $d$-local erasures for $J = sN$.

Our results stem from generalizing and improving  upon results in Refs.~\cite{brandaoQuantum2019,faistContinuous2020,faistTimeenergy2023,liuApproximate2023}.
A special case of our AQECCs is $s=\frac{1}{2}$ and $J=sN=\frac{N}{2}$, in which case $|J\!,\!M\ra$ are the Dicke states~\cite{dickeCoherence1954}, and the corresponding code is the ``thermodynamic code"~\cite{brandaoQuantum2019,faistContinuous2020,faistTimeenergy2023,liuApproximate2023}.
Its code-inaccuracy bound against generic noise on any known $d$ sites was first analyzed in Ref.~\cite{brandaoQuantum2019}, which also reported a coincidental identical bound for the $(s,J)=(1,N)$ code.
We generalize this code to general $s$ and $J \sim N^a$ ($a \leq 1$) and tighten the bound against the same noise. 
We obtain our bound using the SU(2) ladder algebra without explicit expressions for the code words, therefore producing a general $s$- and $J$-dependent bound for this family of codes and explaining the coincidence from Ref.~\cite{brandaoQuantum2019}.

In addition, consider the state $|J\!=\!\frac{N}{2}\!-\!1,M\!=\!J\ra=\frac{1}{\sqrt{N}}\sum_{j=1}^{N}e^{i2\pi j/N}\hat{q}^{-}_j\ket{\uparrow}^{\otimes N}$ along with $|J\!=\!\frac{N}{2}\!-\!1,M\ra$ for the rest of $M$, where $\hat{q}^{z}_j\ket{\uparrow} = \frac{1}{2}\ket{\uparrow}$. 
In Ref.~\cite{gschwendtnerQuantum2019}, a subset of the states with $M \sim J - N^{\kappa}$ and $\kappa<1$ was shown to form an approximate error detection code (``magnon code"). 
Complementing Ref.~\cite{gschwendtnerQuantum2019}, we consider code words $|J\!=\!\frac{N}{2}\!-\!1,M\ra$ with $|M| = O(N^{b})$ for some $b < 1$.

The inaccuracy bound of the ``thermodynamic code" with code words $|M| = O(1)$ in the large-$N$ limit against heralded $d$-local erasure errors for a constant $d$ is reported in Refs.~\cite{faistContinuous2020,liuApproximate2023}.
We generalize it to $s \geq \frac{1}{2}$ and $J=sN$, but with $|M|=O(N^{b})$ and error locality $d\sim N^{c}$ for some exponents $b$ and $c$. 
We obtain a general $s$-dependent code-inaccuracy bound in this case using the Clebsch-Gordan coefficients without writing out the code words.

Our results extend the codes to a regime containing probe states with a metrological entanglement advantage.
The code-inaccuracy bounds also let us bound the loss of QFI. 
Accordingly, we rigorously demonstrate (and numerically confirm) that a noisy probe state in our codes can retain a metrological entanglement advantage under certain conditions.

Our SU(2)-symmetry-irrep-based approach offers a general route for constructing covariant AQECCs with continuous transversal gates, complementing randomized constructions~\cite{kongNearOptimal2022,li2023designs,liSU2023}, reference-frame-assisted constructions~\cite{woodsContinuous2020,haydenError2021,yangOptimal2022}, and an approach which identifies exact codes  from certain discrete-group irreps~\cite{grossDesigning2021,jaincodes,kubischtaNotSoSecret2024,kubischtaQuantum2024,kubischtaQuantum2024a,denysQuantum2024}.
The protection of a metrological entanglement advantage is an additional feature.

\textit{Code words from SU(2) irrep states---}Consider a system of $N$ spin-$s$ degrees of freedom (i.e. qudits of dimension $2s\!+\!1$), where $s$ can be a half integer or an integer.
The total SU(2)-rotation unitary is a tensor product of the SU(2) rotations on the individual spins, generated by the su(2) algebra 
\begin{align}
    \hat{Q}^z \!=\! \sum_{j=1}^{N} \hat{q}^z_j~,~~~
    &\hat{Q}^+ \!=\! \sum_{j=1}^{N} \hat{q}^+_j~,
    &\hat{Q}^-\! = \sum_{j=1}^{N} \hat{q}^-_j ~,
\end{align} 
where $j$ labels the spin and $[\hat{q}_j^+,\hat{q}_j^-]=2\hat{q}^z_j$, $[\hat{q}_j^z,\hat{q}_j^\pm]=\pm\hat{q}_j^\pm$. This yields $[\hat{Q}^+,\hat{Q}^-]=2\hat{Q}^z$ and $[\hat{Q}^z, \hat{Q}^{\pm}]=\pm \hat{Q}^{\pm}$.

The irreps of the total SU(2)-rotation unitary are labeled by $J \leq s N$, with a possible multiplicity, and each irrep has dimension $2J\!+1$ with the basis labeled by $|J\!,\!M\ra$, $M\in\{-J, \cdots, J\}$.
The state $|J\!,\!M\ra$ satisfies
\begin{align}\label{eqn:ladder relation}
    Q^z|J\!,\!M\ra=M|J,M\ra~, ~~Q^{\pm}|J\!,\!M\ra =c_M^{\pm}|J,{M\!\pm\!1}\ra~,
\end{align}
where $c^\pm_M=\sqrt{(J\mp M)(J\pm M+1)}$.

We will show that the code space $\mathfrak{C}=\text{span}\{|J\!,\!M\ra; M\!=\!\Mmin, \Mmin\!+\!\Delta, \Mmin\!+\!2\Delta, \dots, \Mmax\}$, equipped with transversal U(1) logical gates generated by $\hat{Q}^z$, forms an AQECC under generic noise on any known $d$ sites, i.i.d.~error with sufficiently low error probability, and heralded $d$-local erasures for some parameters---provided the spacing $\Delta$ between magnetic quantum numbers is large enough.
For the maximum $J=sN$, $|J\!,\!M\ra$ is totally permutation symmetric, and the irrep multiplicity is one. 
For a general $J < sN$, one can have multiple irreps with the same $J$, and we consider $|J\!,\!M\ra$ in any irrep or from any linear superposition of different irreps with the same $J$, as long as the ladder algebra in Eq.~(\ref{eqn:ladder relation}) holds among the resulting $|J\!,\!M\ra$.
As such, we suppress the multiplicity index.

\textit{Code inaccuracy against generic errors---}We use the notation $(\!(N,k)\!)$ and $\epsilon(\mathcal{N})$  for an AQECC, where $N$ denotes the number of qudits, $k$ denotes the number of logical qubits, and $\epsilon(\mathcal{N})$ denotes the optimal recovery inaccuracy against noise channel $\mathcal{N}$ (with encoding):
\begin{align}
\epsilon(\mathcal{N}) = \min_{\mathcal{R}} D_P(\mathcal{R} \circ \mathcal{N},\mathrm{id})~,
\end{align}
where $\mathcal{R}$ is the recovery channel and $D_P(\mathcal{E},\mathcal{F})$ is the purified distance between channels $\mathcal{E}$ and $\mathcal{F}$, defined as 
\begin{align}
D_P(\mathcal{E},\mathcal{F}) = \max_{|\psi\ra}\sqrt{1- f^2[\mathcal{E}\otimes \mathrm{id}(\psi), \mathcal{F}\otimes \mathrm{id}(\psi)]}~,
\end{align}
where $f(\rho,\sigma):=\tr[\sqrt{\sqrt{\rho}\sigma\sqrt{\rho}}]$ is the fidelity, $\psi:=|\psi\ra\la\psi|$, $|\psi\ra$ is a purification of $\rho$, and $\mathrm{id}$ is the identity channel on the auxiliary space whose input has the same Hilbert space dimension as the input of $\mathcal{E}$ and $\mathcal{F}$ (see Supplemental Material S1~\cite{supp}  for a list of definitions used throughout the paper).
As we only consider noise on at most $d$ sites, we use $\epsilon(d)$ to stress its dependence on $d$, while suppressing $\mathcal{N}$ whose specific choices should be clear from the context. 

We bound the code inaccuracy $\epsilon(\mathcal{N})$ using a generalization~\cite{benyGeneral2010} of the Knill-Laflamme conditions~\cite{knillTheory2000}. If we let $\Pi$ be the projector onto the code space and $\mathcal{N}(\bullet)=\sum_j K_j (\bullet) K_j^{\dagger}$ be the noise channel, then $\epsilon(\mathcal{N}) \leq \delta$ if and only if 
\begin{equation}
\Pi K_i^{\dagger}K_j \Pi = \lambda_{ij}\Pi + \Pi B_{ij}\Pi~,    
\end{equation}
where $\lambda_{ij}$ are the components of a density operator, and $D_P(\Lambda+\mathcal{B},\Lambda)\leq \delta$, where $\Lambda(\rho)=\sum_{i,j}\lambda_{ij}\tr(\rho)|i\ra\la j|$ and $(\Lambda+\mathcal{B})(\rho)=\Lambda(\rho)+\sum_{i,j}\tr(\rho B_{ij})|i\ra\la j|$.

To this end, we start with the following two lemmas. 
\begin{lemma}\label{lemma 1}
    Assume $\hat{F}$ is a $d$-local operator, then $\la J,n|\hat{F}|J,m\ra =0$ if $|n-m|\geq 2sd+1$.
\end{lemma}
Physically, this lemma holds since a $d$-local operator can change the magnetic quantum number $M$ by at most $2sd$.
It ensures the ``off diagonal" part of the Knill-Laflamme conditions to be exactly zero if the code words have enough spacing in $M$.
The proof is simple and given in Supplemental Material S2~\cite{supp}.

To bound the ``diagonal part" of the Knill-Laflamme conditions, we have the following lemma. 
\begin{lemma}\label{lemma 2}
Assume $\hat{F}$ is a $d$-local operator, then $|\la J,n|\hat{F}|J,n\ra-\la J,m|\hat{F}|J,m\ra| \leq dq_0\|\hat{F}\|_{op}C(n,m)$, where $C(n,m):=\sum_{M=n}^{m-1}[c^+_M]^{-1}$ and $q_0$ is a constant. 
In particular, for $0\leq b < a \leq 1$, if we take $J \sim N^a$ and $|m-n|\sim N^{b}$, then we have $| \la J,n|\hat{F}|J,n\ra -\la J,m|\hat{F}|J,m\ra | = \|\hat{F}\|_{op}\cdot O(dN^{-(a-b)})$.
\end{lemma}

Note that $\|\hat{F}\|_{\text{op}} := \sup_{|\psi\ra} \{\|\hat{F}|\psi\ra\|: \||\psi\ra\|\!=\!1 \}$ is the operator norm and $\||\psi\ra\| := \sqrt{\la \psi|\psi\ra}$.
Relegating the proof to Supplemental Material S3~\cite{supp}, we comment that this bound is obtained by bounding $|\la J,n+1|\hat{F}|J,n+1\ra-\la J,n|\hat{F}|J,n\ra|$ and using the triangle inequality. We emphasize that the proof only uses the ladder algebra Eq.~(\ref{eqn:ladder relation}) and the locality of the operator $Q^{\pm}$ without the explicit expression for $|J\!,\!M\ra$.

Using these two lemmas, we obtain the following upper bound on the code inaccuracy for generic noise on any known $d$ sites $\mathcal{N}(\rho)=\sum_{j=1}^{n_d}K_j(\rho)K^{\dagger}_j$ whose Kraus operators $K_j$ are at most $d$ local, acting on a fixed set of $d$ qudits, and $n_d=(2s+1)^{2d}$ is the number of independent Kraus operators. The proof is presented in Supplemental Material S4~\cite{supp} .

\begin{theorem}\label{theorem 1}
    Consider the code space $\mathfrak{C}=\text{span}\{|J\!,\!M\ra; M\!=\!\Mmin, \Mmin\!+\!\Delta, \Mmin\!+\!2\Delta, \dots, \Mmax\}$, where $J\!\sim\! N^a $, $\Mmax \!-\!\Mmin \! \sim \! N^{b}$, and $b\! <\! a\! \leq 1$.  The code $\mathfrak{C}$ forms an $(\!(N,O[b\log_2 N])\!)$ AQECC against $d$-local noise on known sites, with inaccuracy $\epsilon(d) = O[\sqrt{d}(2s+1)^{d} N^{-(a-b)/2}]$ if $\Delta \geq 2sd+1$.
\end{theorem}

\begin{figure}
    \centering
    {\includegraphics[width=0.47\textwidth]{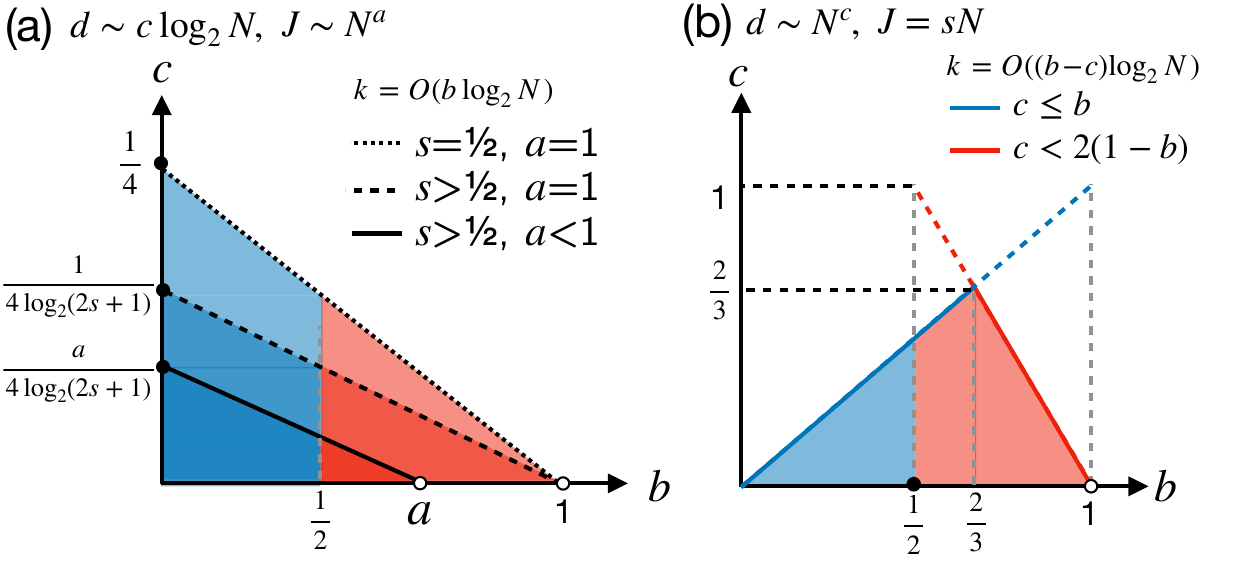}}\\
    \caption{
    (a) For our SU(2) code with total spin $J\! \sim\! N^a$ and magnetic quantum number range $\Mmax\! -\! \Mmin \sim N^b$, under heralded $d$-local noise channel ($d\! \sim \! c \log_2 N$), the $(b,c)$ region where $\epsilon\! \to\! 0$ (blue and red areas), and the region with metrological advantage (red area). (b) For $J = sN$, $\Mmax - \Mmin \sim N^b$, and $d$-local erasures ($d \sim N^c$), the $(b,c)$ region where $\epsilon \to 0$ (blue and red areas), and the region with metrological advantage (red area), all independent of $s$.}
    \label{fig:parameter}
\end{figure}

We focus on parameters where $\epsilon \rightarrow 0$ in the thermodynamic limit $N \rightarrow \infty$.
For error size $d\! \sim \! c \log_2N$, we have $\epsilon = O(\sqrt{\log_2N} N^{-[a-b-2c\log_2(2s+1)]/2})$, which vanishes as $N \rightarrow \infty$ provided $c < (a-b)/[4\log_2(2s+1)]$. This is illustrated in Fig.~\ref{fig:parameter}(a) as the shaded area.

Our code-inaccuracy bounds generalize and improve the results in Ref.~\cite{brandaoQuantum2019}, which gives $\epsilon'=\tilde{O}(n_dN^{(5b-1)/2})$ for special cases $(s,J)=(\frac{1}{2}, \frac{N}{2})$ and $(s,J)=(1, N)$, while we achieve $\epsilon = \tilde{O}(\sqrt{n_d} N^{(b-1)/2})$.
Our improvements from $N^{5b}$ to $N^{b}$ and from $n_d$ to $\sqrt{n_d}$ come from setting the ``off-diagonal" Knill-Laflamme conditions to zero and a different proof technique.

Moreover, the derivation of the bounds in Ref.~\cite{brandaoQuantum2019} relies on the $b\leq \frac{1}{2}$ assumption; our approach allows $1 > b > \frac{1}{2}$, which enables a probe state with a metrological entanglement advantage, as we will discuss below. 
Notably, using the ladder algebra of SU(2), we obtain a general $s$- and $J$-dependent bound on the code inaccuracy, explaining the coincidence of code-inaccuracy bounds for $(s,J)=(\frac{1}{2},\frac{N}{2})$ and $(s,J)=(1,N)$ in Ref.~\cite{brandaoQuantum2019}.

The reason to consider the noise on a fixed set of $d$ sites in Theorem~\ref{theorem 1} is to control the number of Kraus operators $n_d$. For generic $d$-local noise $\mathcal{N}(\rho)=\sum_{j=1}^{n_d}K_j(\rho)K^{\dagger}_j$, where $K_j$ can act on \textit{different} sets of $d$ qudits, $\mathfrak{C}$ forms an $(\!(N, O[b\log_2 N])\!)$ AQECC with $\epsilon(d)=O(\sqrt{d}n_d N^{-(a-b)/2})$ if we take $\Delta \geq 4sd+1$. (See Supplemental Material~\cite{supp} S4 for Corollary 1.)

Physically, the above result implies that the code can tolerate any error with the number of Kraus operators growing slower than $O(\sqrt{d}N^{-(a-b)/2})$, assuming each has operator norm $O(1)$. 
However, in practice, more nonlocal errors tend to occur with lower probability, resulting in smaller operator norms for their corresponding Kraus operators. 
This structure allows the code to tolerate noise channels with a larger number of Kraus operators.
Indeed, under i.i.d.~errors at unknown locations, we show in Supplemental Material~\cite{supp} S5 (Corollary 2) that our code can achieve vanishing code inaccuracy by correcting $d$-local errors with $d \sim c\log_2 N$, provided the error rate is sufficiently low $p \lesssim N^{-2} \log^2 N$.
Although this error rate implies a vanishing expected number of errors $Np \lesssim N^{-1} \log^2 N$, it is chosen to enable a rigorous analytical bound on the code inaccuracy.
It remains an interesting open question whether the code can still achieve low inaccuracy when the number of errors becomes nonvanishing.

\textit{Code inaccuracy against heralded erasures---}The inaccuracy of the ``thermodynamic code"  against heralded $d$-local erasures for $d=O(1)$ and $|M|= O(1)$ was found in Refs.~\cite{faistContinuous2020,liuApproximate2023} to be $\epsilon = O(N^{-1})$.
The heralded $d$-local erasure channel is $\mathcal{N}_{\alpha}(\rho):= \tr_{\alpha}[\rho]$, where $\alpha$ denotes the set of $d$ erased qudits.

We now generalize the code-inaccuracy bound against heralded erasures to $s$, $d \sim N^{c}$ and $|M| = O (N^{b})$. 
In this case, we consider codes with $J=sN$ only, where $|J\!,\!M\ra$ is totally permutation symmetric.
Owing to the irrep multiplicity being one, the $J = sN$ assumption allows us to express each codeword uniquely in terms of a basis for two subsystems using Clebsch–Gordan coefficients.
To bound the code inaccuracy, we follow Ref.~\cite{benyGeneral2010} by bounding the channel distance via the complementary channel $\hat{\mathcal{N}}_{\alpha}(\bullet)=\tr_{\bar{\alpha}}(\bullet)$, where $\bar{\alpha}$ denotes the set of qudits not in $\alpha$ (the complement of $\alpha$). This can be achieved by the following lemma.

\begin{lemma}\label{lemma 3}
Consider the irrep of SU(2) formed by $|J\!,\!M\ra$, where $J=sN$, and the reduced density matrix $\rho_M := \tr_{\bar{\alpha}}[|J\!,\!M\ra \la J\!,\!M|]$, where $\bar{\alpha}$ is the complement of the set of erased qudits $\alpha$. 
We have the fidelity $f(\rho_M,\rho_{M=0})= 1-O(\frac{d M^2}{sN^2})= 1-O(N^{2b+c-2})$ asymptotically for $M \sim N^b$, $d \sim N^c$, and $1 > b \geq c$. 
\end{lemma}

Physically, the above lemma states that the environment can barely tell the code words apart from each other.
In the proof (presented in Supplemental Material S6~\cite{supp}), we use Clebsch-Gordan coefficients without writing the code words out explicitly, enabling a general $s$-dependent bound on the fidelity for $J\!=\!sN$.
In particular, we can rewrite $|J\!,\!M\ra = \sum_{m_1=-j_1}^{j_1} C^{JM}_{j_1,m_1,j_2,m_2} |j_1\!,\!m_1\ra_{\alpha} \otimes |j_2\!,\!m_2\ra_{\bar{\alpha}}$, where $m_2=M-m_1$, $j_1 =sd$, $j_2=J-j_1$, and $C^{JM}_{j_1,m_1,j_2,m_2}= (\la j_1,m_1|\la j_2,m_2|)|J,M\ra$ is the Clebsch-Gordan coefficient.
The fidelity between $\rho_M$ and $\rho_{M=0}$ can therefore be calculated accordingly.

Our proof for a general $s$ has essentially the same asymptotic scaling as the proof of Lemma 13 in Ref.~\cite{faistContinuous2020} for $s=\frac{1}{2}$, which reported $f(\rho_M,\rho_{M=0})= 1 - D_{d,M}/N^2 + O(N^{-3})$, where $D_{d,M} = d^2/8+[d(1+M^2+2M)+2d^3+d^2(2M-1)]/4$.
This indeed becomes our result in the asymptotic regime $M \sim N^b$, $d \sim N^c$, and $1 > b \geq c$, making the leading-order term $\sim dM^2$.

The following theorem bounds the code inaccuracy against a heralded $d$-local erasure error.
\begin{theorem}\label{theorem 2}
    Consider the code space $\mathfrak{C}=\text{span}\{|J\!,\!M\ra; M\!=\!\Mmin, \Mmin\!+\!\Delta, \Mmin\!+\!2\Delta, \dots, \Mmax\}$, where $J=sN$, $\Mmax-\Mmin \sim N^{b}$ and $b < 1$.  $\mathfrak{C}$ forms an $(\!(N,O[(b\!-\!c)\log_2N])\!)$ AQECC against heralded $d$-local erasures of size $d \sim N^{c}$, where $0 \leq c \leq b < 1$, with $\epsilon(d)=O(N^{-(1-b-c/2)})$ if we take $\Delta \geq 2sd +1$.
\end{theorem}
Theorem~\ref{theorem 2} can be obtained by combining our Lemma~\ref{lemma 1}, our Lemma~\ref{lemma 3}, and Theorem 3 of Ref.~\cite{faistContinuous2020}.
To make our presentation self-contained, we present the proof of Theorem~\ref{theorem 2} tailored to our case in Supplemental Material S7~\cite{supp}.
References~\cite{faistContinuous2020,liuApproximate2023} showed that $\epsilon = \Omega(N^{-1})$ for any U(1) covariant code. 
Thus, this family of SU(2) AQECCs is asymptotically optimal for general $s$ since $\epsilon = O(N^{-1})$ when $b=c=0$.
Although Theorem~\ref{theorem 1} also applies to heralded erasure errors, the improvement in Theorem 2 arises from the assumption that $J = sN$, enabling a stronger result.

For general $b$ and $c$, to have $\epsilon \rightarrow 0$ when $N \rightarrow \infty$, we need $c < 2(1-b)$. 
Additionally, we need $\Delta = \Omega(N^c)$ for Lemma~\ref{lemma 1} and $c \leq b$ to ensure at least one logical qubit. This parameter regime is shaded in Fig.~\ref{fig:parameter}(b), which is independent of $s$ and $k$.
The code words with $M \sim N^b$ for $b > \frac{1}{2}$ [red shaded area in Fig.~\ref{fig:parameter}(b)] are crucial for a metrological entanglement advantage.

\textit{Metrological entanglement advantage---}We now apply our codes to quantum metrology and show that they can have an entanglement advantage with noise resistance.
In the local parameter estimation problem~\cite{faistTimeenergy2023}, we would like to estimate the parameter $\theta$ imprinted on the probe state $|\psi_0\ra \in \mathfrak{C} $ as $|\psi_\theta\ra=e^{-i\hat{Q}^z \theta} |\psi_0\ra \in \mathfrak{C}$. 
Abbreviating $\psi_{\theta} :=|\psi_\theta\ra \la \psi_\theta|$, the noisy probe state is then $\rho_{\theta} := \mathcal{N}(\psi_{\theta})$ for the noise channel $\mathcal{N}$.

In quantum metrology, QFI is a central figure of merit, which characterizes how well a parameter can be estimated via the quantum Cram\'er-Rao bound~\cite{braunsteinStatistical1994}.
For the states $\rho_{\theta}$ parametrized by $\theta$, the QFI is defined as
\begin{equation}
    \mathfrak{F}(\rho_{\theta};\partial_{\theta}\rho_{\theta}):= \tr(\rho_{\theta}R^2)~,
\end{equation}
where $R$ satisfies $\rho_{\theta} R+R\rho_{\theta}=2\partial_{\theta}\rho_{\theta}$.

If $\theta$ is generated from $e^{-i \theta \hat{Q}^z}$, then without entanglement, the QFI is upper bounded by the standard quantum limit, scaling as $\mathfrak{F}\! \sim\! N$.
Achieving $\mathfrak{F} \!\sim\! N^{\gamma}$ for $1\! <\! \gamma\! \leq \! 2$ indicates a metrological entanglement advantage, with $\gamma\! = \!2$ as the Heisenberg limit. 
Since $\mathfrak{F}(\psi_{\theta};\partial_{\theta}\psi_{\theta})= 4[\la \psi_\theta|(\hat{Q}^z)^2|\psi_\theta \ra-\la \psi_\theta| \hat{Q}^z|\psi_\theta \ra^2]$ for a pure state, our code can have an entanglement advantage if it contains a logical state whose variance of $\hat{Q}^z$ scales superlinearly with $N$. 
This is possible with the probe state $|\psi_0\ra \propto |J,M\ra + |J,-M\ra$, whose QFI is $\mathfrak{F}(\psi_{\theta})=4M^2\! \sim\! N^{2b}$ for $M \!\sim \!N^{b}$ and $b>\frac{1}{2}$.

While a GHZ(-like) state $|\psi'\ra=(|J,J\ra + |J,-J\ra)/\sqrt{2}$ can achieve the Heisenberg limit~\cite{giovannettiAdvances2011,RevModPhys.89.035002,pezzeQuantum2018a}, its QFI becomes zero even if only one qudit is erased (see Supplemental Material~\cite{supp} S8). 
Various other probe states also provide an entanglement advantage, e.g., spin-squeezed \cite{maQuantum2011a,frankeQuantumenhanced2023,bornetScalable2023,hinesSpin2023,ecknerRealizing2023}, Dicke \cite{pezzeQuantum2018a,10.3389/fphy.2024.1369786,saleemAchieving2024}, and scrambled \cite{liImproving2023a} states.
Additionally, various studies have employed different methods to assess the effect noise in metrology~\cite{escherGeneral2011,demkowicz-dobrzanskielusive2012,kolodynskiEfficient2013,chavesNoisy2013,demkowicz-dobrzanskiAdaptive2017,zhouAsymptotic2021,yinHeisenberglimited2023}.
Differently, we use AQECC inaccuracy to rigorously bound the QFI loss under certain noise, showing that entanglement advantage persists.  
Notably, the QFI protection here is achieved without repeated quantum error correction~\cite{faistTimeenergy2023}, differing from the typical error-corrected metrology scenarios~\cite{kesslerQuantum2014,arradIncreasing2014,durImproved2014,luRobust2015a,matsuzakiMagneticfield2017,sekatskiQuantum2017,zhouAchieving2018,zhouOptimal2020,rojkovBias2022}.
In a recent work~\cite{yinSmall2024}, a nonmaximally symmetric state also achieves metrological entanglement advantage without error correction.

To this end, we use the following theorem adapted from Ref.~\cite{faistTimeenergy2023} to bound the loss of QFI.

\begin{customthm}{27}[Adapted from Ref.~\cite{faistTimeenergy2023}]\label{prop27}
    Let $\psi_\theta$ be a pure state and $\partial_\theta\psi_\theta = -i[\hat{Q}^z,\psi_\theta]$. If, for $\epsilon > 0$, $\mathcal{M}$ is a channel with a diamond distance from identity satisfying $\|\mathcal{M}-\mathrm{id}\|_{\diamond}\leq 2\epsilon$, then 
    \begin{equation}\label{eqn:QFI loss bound}
         \mathfrak{F}(\psi_\theta;\partial_\theta \psi_\theta)-\mathfrak{F}[\mathcal{M}(\psi_\theta);\mathcal{M}(\partial_\theta \psi)] \leq  16\epsilon \|\partial_\theta \psi_\theta\|_{1}\|\partial_\theta \psi_\theta\|_{\text{op}}~.
    \end{equation}
\end{customthm}
Since $\|\mathcal{E}-\mathcal{F}\|_{\diamond} \leq 2D_P(\mathcal{E},\mathcal{F})$ (see Supplemental Material S1~\cite{supp}), taking $\mathcal{M}=\mathcal{R}\circ\mathcal{N}$, we have $\|\mathcal{M}-\mathrm{id}\|_{\diamond} \leq 2\epsilon$ bounded by the AQECC inaccuracy.
In fact, from the data processing inequality~\cite{ferrieDataprocessing2014,faistTimeenergy2023}, we have $\mathfrak{F}[\mathcal{R}\circ \mathcal{N}(\rho_\theta);\mathcal{R}\circ \mathcal{N}(\partial_{\theta}\rho_\theta)] \leq \mathfrak{F}[\mathcal{N}(\rho_\theta);\mathcal{N}(\partial_{\theta}\rho_\theta)]$.
That is, the protection of QFI can already be achieved by encoding alone.

For metrology purposes, $\mathfrak{C}=\text{span}\{|J\!,\!M\ra, |J\!,\!-M\ra\}$ suffices to host the probe state $|\psi_0\ra = (|J\!,\!M\ra + |J\!,\!-M\ra)/\sqrt{2}$, which has $\|\partial_{\theta}\psi_\theta\|_1=\|\partial_{\theta}\psi_\theta\|_{\text{op}}=M$. 
Therefore, we have $\mathfrak{F}[\mathcal{N}(\psi_{\theta});\mathcal{N}(\partial_\theta \psi_{\theta})] \geq (1-4\epsilon)4M^2$, and a protected entanglement advantage is possible if $\epsilon \rightarrow 0$ when $N \rightarrow \infty$ while having $M\sim N^b$ with $b > 1/2$. 

Now considering the $d$-local noise, for $J\! \sim\! N^a$, $M\! \sim\! N^b$, and $d\! \sim\! c \log_2(N)$, Theorem~\ref{theorem 1} gives $\epsilon = \tilde{O}[N^{-(a-b-2c\log_2(2s+1))/2}]$. So for $a > b > 1/2$, the noisy probe state can maintain an entanglement advantage $\mathfrak{F}[\mathcal{N}(\psi_\theta)] = \Theta(N^{2b}) $ if $c < (a-b)/[4\log_2(2s+1)]$ [illustrated in Fig.~\ref{fig:parameter}(a)]. 
For heralded $d$-local erasure errors with $d\sim N^{c}$, Theorem~\ref{theorem 2} gives $\epsilon=O(N^{b+c/2-1})$. 
The partially-erased probe state still has an entanglement advantage $\mathfrak{F}[\mathcal{N}(\psi_\theta)] =\Theta(N^{2b})$ with $1 > b > 1/2$ and $\epsilon \rightarrow 0$ provided  $c < 2(1-b)$ and $c \leq b$ [illustrated in Fig.~\ref{fig:parameter}(b)].

\begin{figure}
    \centering
    {\includegraphics[width=0.48\textwidth]{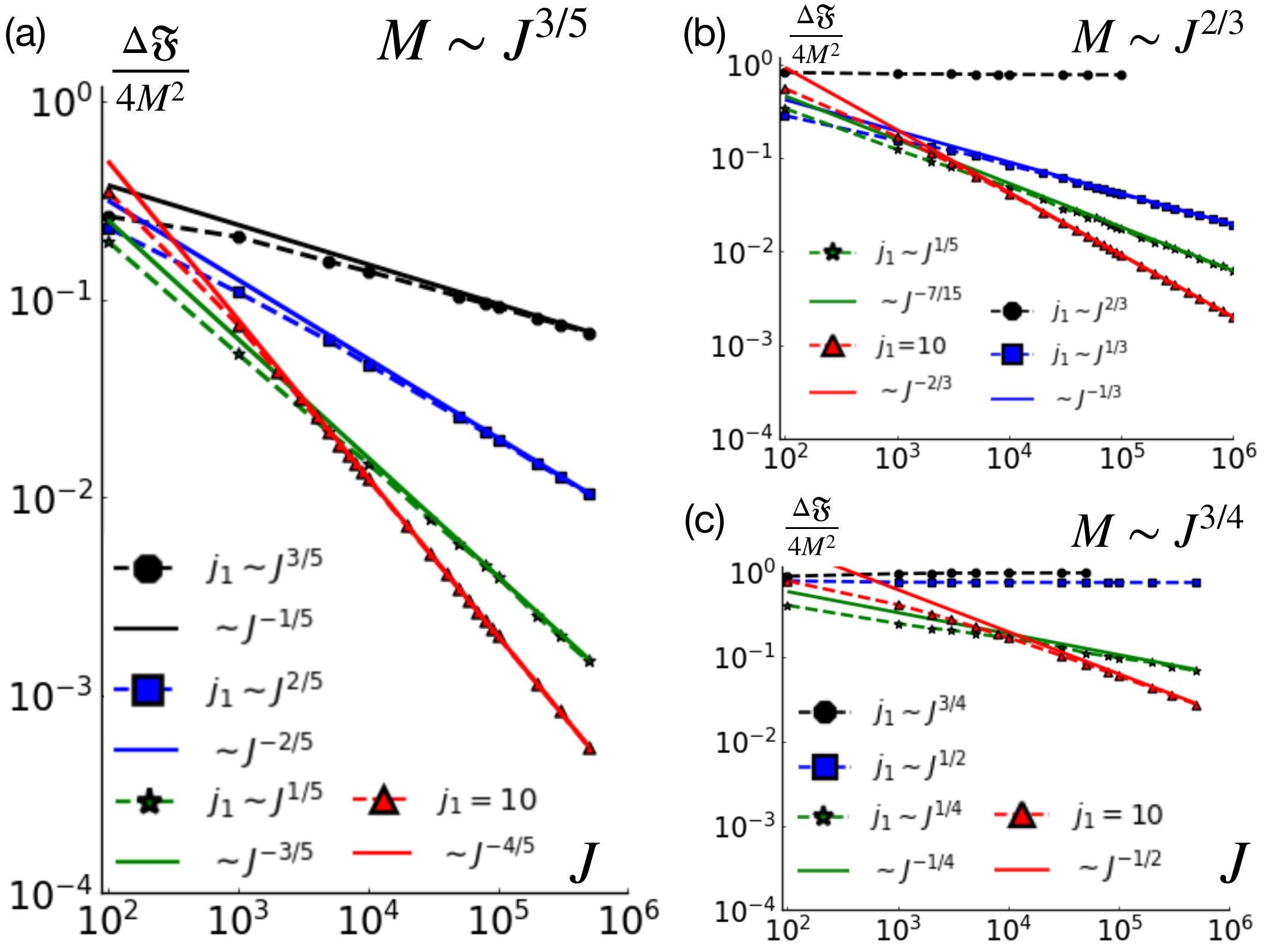}}\\
    \caption{QFI loss $\Delta \mathfrak{F}$ of the probe state $|\psi \ra \propto |J,M\ra+|J,-M\ra$ ($J=sN$) for (a) $M \sim J^{\frac{3}{5}} $, (b) $M \sim J^{\frac{2}{3}} $, and (c) $M \sim J^{\frac{3}{4}} $, after erasing $d=j_1/s$ qudits. 
    We calculate $\Delta \mathfrak{F}$ by picking values of $J$ and $j_1$, giving us results for any spin $s$ and corresponding $N\!=\!J/s$ and $d=j_1/s$.
    Numerical results show $\Delta \mathfrak{F} /(4M^2) \sim J^{2b+c-2}$ for the parameters in the red-shaded region of Fig.~\ref{fig:parameter}(b).}

    \label{fig:QFIloss}
\end{figure}

To assess our QFI-loss bound, we numerically calculate  $\Delta \mathfrak{F} := \mathfrak{F}(\psi_\theta)-\mathfrak{F}[\mathcal{N}(\psi_\theta)]$ for $J=sN$, where the channel is $\mathcal{N}(\psi_{\theta})=\tr_\alpha[\psi_{\theta}]$ and $\alpha$ contains $d$ spins (so $j_1=sd$).
In Fig.~\ref{fig:QFIloss}, we plot $\Delta \mathfrak{F}/(4M^2)$ for various $b$ and $c$ in the red-shaded region of Fig.~\ref{fig:parameter}(b). The numerical calculations are detailed in Supplemental Material S8~\cite{supp}.
Applying the QFI-loss bound Eq.~(\ref{eqn:QFI loss bound}), we have $\Delta \mathfrak{F}/(4M^2) \leq 4\epsilon$. 
We observe that $\Delta \mathfrak{F}/M^2 \sim J^{2b+c-2}$, which is indeed asymptotically upper bounded by $J^{b+c/2-1} \sim N^{b+c/2-1}$ .
Though the bound is not tight, it could be due to the Fuchs-van de Graaf inequality~\cite{fuchsCryptographic1999} not being saturated asymptotically.
However, within the red region in Fig.~\ref{fig:parameter}(b), $\mathfrak{F}(\psi_\theta)$ and $\mathfrak{F}[\mathcal{N}(\psi_\theta)]$ both scale as $N^{2b}$, both indicating a metrological entanglement advantage.

In Supplemental Material S8~\cite{supp}, we present a local measurement that saturates the quantum Cram\'er-Rao bound for the ideal probe state $|\psi_\theta\ra$. For the partially erased probe state, a global measurement can achieve the QFI scaling, while a local measurement can surpass the standard limit in certain parameter regimes, despite not saturating the quantum Cram\'er-Rao bound.

\textit{Discussion---}We establish code-inaccuracy bounds for a family of covariant approximate quantum error-correcting codes whose codewords house irreducible representations of SU(2).
We show that a particular (entangled) codeword can be used for quantum sensing in a way that outperforms classical strategies, even in the presence of certain noise.

References~\cite{bartschiDeterministic2019,nepomechieSpins$s$2024,nepomechieQudit2024,yuEfficient2024,liuLowdepth2025} show protocols to prepare $|J\!,\!M\ra$ states for $J=sN$. 
One can then prepare the probe state $|\psi_0\ra$ using an ancilla qubit and a controlled version of the circuit that prepares $|J\!,\!M\ra$.

We have proposed some measurement schemes for the ideal and partially-erased probe states in Supplemental Material S8~\cite{supp}.
It is worth investigating whether simpler measurements can saturate the Cram\'er-Rao bound for our noisy state~\cite{zhouSaturating2020a}. While the recovery channel~\cite{benyGeneral2010,ngsimple2010} cannot increase quantum Fisher information~\cite{ferrieDataprocessing2014,faistTimeenergy2023}, it may aid in devising a measurement protocol for any noisy probe.
Refs.~\cite{benyGeneral2010,benyApproximate2011} proposed a near-optimal recovery channel for AQECCs, which we outline in Supplemental Material S9~\cite{supp}. 
Although an efficient circuit implementation is not guaranteed, it would be interesting to explore whether the structure of the $J = sN$ code, such as its total permutation symmetry, can be leveraged for this purpose. 

With modifications, our approximate quantum error-correcting code results can be applied to other SU(2) code words, like the quantum many-body scar states in the spin-1 XY-Dzyaloshinskii-Moriya model~\cite{schecterWeak2019,markUnified,dooleyRobust2021} (see Supplemental Material S10~\cite{supp} for details). 
Accordingly, the SU(2)-scar probe states could be noise-resilient, complementing the results using scar states for sensing~\cite{dooleyRobust2021,dooleyEntanglement2023,yoshinagaQuantum2022}.
For states with approximate SU(2) symmetry~\cite{liuQuantum2022,liuApproximate2023}, like the scars in a deformed PXP model~\cite{bernienProbing2017,turnerWeak2018,turnerQuantum2018,linExact2019,choiEmergent2019}, some exact codes~\cite{liuQuantum2022,liuApproximate2023}, or approximately prepared states~\cite{piroliApproximating2024a}, further study is warranted to determine if they also provide a noise resilient entanglement advantage.

Our result suggests a new way of finding approximate quantum error-correcting codes with continuous transversal logical gates by using symmetry irreps, complementing existing methods~\cite{woodsContinuous2020,haydenError2021,yangOptimal2022,kongNearOptimal2022,liSU2023}. 
Here, one selects a subset of the SU(2)-irrep basis to form an approximate quantum error-correcting code, reducing the symmetry to U(1) and giving rise to continuous transversal logical gates.
This raises the question of whether one can also devise approximate quantum error-correcting codes by using irreps from other continuous symmetries [say, SU(3)], which could potentially yield transversal gates generated by multiple (possibly noncommuting) generators and an entanglement advantage for multiparameter sensing.

\begin{acknowledgments}
We thank Xiaozhen Fu, Shubham Jain, Zhi Li, Zhenning Liu, Sean~R.~Muleady, Nicole Yunger Halpern, and Sisi Zhou for valuable discussions and feedback. 
C.-J.~Lin acknowledges the support from the National Science Foundation (QLCI Grant No.~OMA-2120757).
Z.-W.L. is supported in part by a startup funding from YMSC, Tsinghua University, Dushi Program, and NSFC under Grant No.~12475023.
A.V.G.~was supported in part by AFOSR MURI, DARPA SAVaNT ADVENT, NSF QLCI (Award No.~OMA-2120757), the DoE ASCR Quantum Testbed Pathfinder program (Awards No.~DE-SC0019040 and No.~DE-SC0024220), NSF STAQ program, ARL (W911NF-24-2-0107), and NQVL:QSTD:Pilot:FTL. A.V.G.~also acknowledges support from the U.S.~Department of Energy, Office of Science, National Quantum Information Science Research Centers, Quantum Systems Accelerator (QSA) and from the U.S.~Department of Energy, Office of Science, Accelerated Research in Quantum Computing, Fundamental Algorithmic Research toward Quantum Utility (FAR-Qu).  
\end{acknowledgments}

%

\include{supplementary.tex}


\end{document}

%% file: supplementary.tex
\clearpage

\setcounter{equation}{0}
\setcounter{figure}{0}
\setcounter{table}{0}
\setcounter{page}{1}
\makeatletter
\renewcommand{\thesection}{S\arabic{section}}
\renewcommand{\theequation}{S\arabic{equation}}
\renewcommand{\thefigure}{S\arabic{figure}}


\onecolumngrid
\begin{center}
\textbf{\large Supplemental Material for ``Covariant Quantum Error-correcting Codes with Metrological Entanglement Advantage''}
\\~\\
Cheng-Ju Lin$^{1,2}$, Zi-Wen Liu$^{3,4}$, Victor V. Albert$^{1}$, Alexey V. Gorshkov$^{1,2}$ \\
\vspace{.05in}
\small{
$^{1}$\textit{Joint Center for Quantum Information and Computer Science,\\ NIST/University of Maryland, College Park, Maryland 20742, USA}\\
$^{2}$\textit{Joint Quantum Institute, NIST/University of Maryland, College Park, Maryland 20742, USA}\\
$^{3}$\textit{Yau Mathematical Sciences Center, Tsinghua University, Beijing 100084, China}\\
$^{4}$\textit{Perimeter Institute for Theoretical Physics, Waterloo, Ontario N2L 2Y5, Canada}
}
\end{center}

\onecolumngrid
This Supplemental Material is organized as follows. In Sec.\ \ref{app:list of def}, we list the definitions of various quantities used throughout the paper. We then present the proofs of Lemma 1 (Sec.\ \ref{app:proof lemma 1}), Lemma 2 (Sec.\ \ref{app:proof lemma 2}), Theorem 1 \& Corollary 1 (Sec.\ \ref{app:proof theorem 1}), Corollary 2 (Sec.\ \ref{app:proof Corollary 2}), Lemma 3 (Sec.\ \ref{app:proof lemma 3}), and Theorem 2 (Sec.\ \ref{app:proof theorem 2}). In particular, Lemma 1 pertains to the ``off-diagonal" part of the Knill-Laflamme conditions, while Lemma 2 bounds the ``diagonal" part. 
Theorem 1 and Corollary 1 give a bound on the code inaccuracy against a generic $d$-local noise channel. 
Corollary 2 gives conditions and an explicit bound on the code inaccuracy under independent and identically distributed (i.i.d.) noise.
Lemma 3 pertains to the fidelity between the reduced density matrices from code words of different magnetic quantum numbers.
Theorem 2 describes an upper bound on the code inaccuracy against a heralded $d$-local erasure error.
In Sec.\ \ref{app:QFIloss}, we provide the details behind the numerical calculation of the quantum Fisher information and behind the measurement protocols described in the main text.
We discuss a construction of a recovery channel in Sec.\ \ref{app:recovery}.
Finally, in Sec.\ \ref{app:scars}, we show that the quantum many-body scar states in the spin-1 XY-Dzyaloshinskii-Moriya model can also form an covariant AQECC. 

\section{A List of definitions\label{app:list of def}}
In this section, we list the definitions used throughout the paper.
\begin{definition}{(Operator norm)}
     $\|\hat{F}\|_{\text{op}} := \sup_{|\psi\ra} \{\|\hat{F}|\psi\ra\|: \||\psi\ra\|\!=\!1 \}$
\end{definition}
is the operator norm, where $\||\psi\ra\| := \sqrt{\la \psi|\psi\ra}$ is the vector norm.

\begin{definition}{(Trace norm)}
     $\|\rho\|_1 := \tr \sqrt{\rho^{\dagger}\rho}$~,
\end{definition}
\noindent which induces the trace distance between two density matrices $\rho$ and $\Gamma$ as $\|\rho-\Gamma \|_1$.
The trace distance then induces a distance measure between quantum channels $\mathcal{E}$ and $\mathcal{F}$:
\begin{definition}{(Diamond distance)}
     $\|\mathcal{E}-\mathcal{F}\|_\diamond := \max_{|\psi \ra} \| \mathcal{E}\otimes \mathrm{id}(|\psi \ra \la \psi| ) - \mathcal{F}\otimes \mathrm{id} (|\psi \ra \la \psi|)  \|_1$~, 
\end{definition}
\noindent  where $|\psi \ra$ is a purification of $\rho$ and $\mathrm{id}$ is the identity channel on the auxiliary space whose input has the same Hilbert space dimension as the input of $\mathcal{E}$ and $\mathcal{F}$.

\begin{definition}{(Fidelity)}
     $f(\rho,\sigma) := \|\sqrt{\sigma}\sqrt{\rho}\|_1 = \tr \sqrt{\sqrt{\rho}\sigma \sqrt{\rho}}$~,
\end{definition}

\noindent  which induces the entanglement fidelity between channels $\mathcal{E}$ and $\mathcal{F}$: 

\begin{definition}{(Entanglement fidelity)}
     $F(\mathcal{E},\mathcal{F}) := \min_{|\psi \ra} f(\mathcal{E}\otimes \mathrm{id}(|\psi \ra \la \psi|), \mathcal{F}\otimes \mathrm{id}(|\psi \ra \la \psi|))$~,
\end{definition}
\noindent where again $|\psi \ra$ is a purification of $\rho$ and $\mathrm{id}$ is the identity channel on the auxiliary space whose input has the same Hilbert space dimension as the input of $\mathcal{E}$ and $\mathcal{F}$.
The purified channel distance is defined as follows:
\begin{definition}{(Purified channel distance)}
     $D_P(\mathcal{E},\mathcal{F}) := \sqrt{1-F(\mathcal{E},\mathcal{F})^2}$~.
\end{definition}
Note that one can also define another channel distance measure, called Bures distance as follows:
\begin{definition}{(Bures channel distance)}
     $D_B(\mathcal{E},\mathcal{F}) := \sqrt{1-F(\mathcal{E},\mathcal{F})}$~.
\end{definition}

In Ref.~\cite{benyGeneral2010}, the code inaccuracy is quantified by the Bures distance. On the other hand, in this work, we opt for using purified distance. These channel distances are related, however, via the following Fuchs-van de Graaf inequality~\cite{fuchsCryptographic1999}:
\begin{equation}
\label{eqn: fuchs-van de Graaf}
    1-f(\rho,\sigma) \leq \frac{1}{2}\|\rho - \sigma \|_1 \leq \sqrt{1-f(\rho,\sigma)^2}~.
\end{equation}
This immediately gives us 
\begin{equation}\label{eqn:channel distance relations}
    D_B(\mathcal{E},\mathcal{F})^2 \leq \frac{1}{2} \| \mathcal{E}-\mathcal{F}\|_{\diamond} \leq D_P(\mathcal{E},\mathcal{F})~.
\end{equation}

\section{Proof of Lemma 1}\label{app:proof lemma 1}

In this section, we present the proof of Lemma 1, which pertains to the ``off-diagonal" part of the Knill-Laflamme conditions.

\textit{Lemma 1: Assume $\hat{F}$ is a $d$-local operator, then $\la J,n|\hat{F}|J,m\ra =0$ if $|n-m|\geq 2sd+1$.}
\\
\begin{proof} 
Recall that $\hat{Q}^z|J,n\ra = n |J,n\ra$, where $\hat{Q}^z=\sum_{j=1}^{N}\hat{q}^z_j$. Since we consider a system of spin-$s$ degrees of freedom, the charge $\hat{q}^z$ can have eigenvalues $q^z= -s,-s+1, \cdots, s$. Therefore, a $d$-local operator can change the total charge by at most $2sd$. So we have $\la J,n|\hat{F}|J,m\ra =0$, if $|n-m|\geq 2sd+1$.

\end{proof}

\section{Proof of Lemma 2}\label{app:proof lemma 2}
In this section, we present the proof of Lemma 2, which bounds the ``diagonal" part of the Knill-Laflamme conditions.

\textit{Lemma 2: Assume $\hat{F}$ is a $d$-local operator, then $|\la J,n|\hat{F}|J,n\ra-\la J,m|\hat{F}|J,m\ra| \leq dq_0\|\hat{F}\|_{op}C(n,m)$, where $C(n,m):=\sum_{M=n}^{m-1}[c^+_M]^{-1}$ and $q_0$ is a constant. 
In particular, for $0\leq b < a \leq 1$, if we take $J \sim N^a$ and $|m-n|\sim N^{b}$, then we have $| \la J,n|\hat{F}|J,n\ra -\la J,m|\hat{F}|J,m\ra | = \|\hat{F}\|_{op}\cdot O(dN^{-(a-b)})$. }

\begin{proof} 
Assume $I$ is the set of qudits $\hat{F}$ acts on. We then have 
\begin{equation}
    \|[\hat{F},\hat{Q}^+]\|_{op} \leq \|[\hat{F},\sum_{j \in I} \hat{q}^{+}_{j}]\|_{op} \leq d q_0\|\hat{F}\|_{op}~,
\end{equation}
where we denoted $q_0=2 \|\hat{q}^+_{j}\|_{op}$ for convenience, which is a constant. 
To bound $|\la J,{n+1}|\hat{F}|J,{n+1}\ra-\la J,n|\hat{F}|J,n\ra|$, note that 
\begin{equation}
    \la J,{n+1}|\hat{F}|J,{n+1}\ra = \frac{1}{(c^{+}_n)^2}\la J,{n}|\hat{Q}^-  \hat{F} \hat{Q}^+|J,{n}\ra =  \frac{1}{(c^{+}_n)^2}\la J,{n}|\hat{Q}^- \hat{Q}^+ \hat{F}|J,{n}\ra + \frac{1}{c^{+}_n}\la J,{n+1}|[\hat{F},\hat{Q}^+]|J,{n}\ra~.
\end{equation}
However, recall $c_n^{\pm}=\sqrt{(J\mp n)(J\pm n +1)}$, we also have 
\begin{equation}
    \frac{1}{(c^{+}_n)^2}\la J,{n}|\hat{Q}^- \hat{Q}^+ \hat{F}|J,{n}\ra=\frac{c^{-}_{n+1}}{c^{+}_n}\la J,{n}|\hat{F}|J,{n}\ra=\la J,{n}|\hat{F}|J,{n}\ra~,
\end{equation}
which gives us 
\begin{equation}
    |\la J,{n+1}|\hat{F}|J,{n+1}\ra-\la J,n|\hat{F}|J,n\ra|\leq dq_0\|\hat{F}\|_{op} \cdot (c^{+}_n)^{-1}~,
\end{equation}
or 
\begin{align}\label{eqn: Jdiff in C}
    |\la J,{m}|\hat{F}|J,{m}\ra-\la J,n|\hat{F}|J,n\ra| &= \left|\sum_{\ell=n}^{m-1} \la J,{\ell+1}|\hat{F}|J,{\ell+1}\ra-\la J,\ell|\hat{F}|J,\ell\ra \right| \notag \\
    &\leq \sum_{\ell=n}^{m-1}| \la J,{\ell+1}|\hat{F}|J,{\ell+1}\ra-\la J,\ell|\hat{F}|J,\ell\ra | \notag \\
    &\leq dq_0\|\hat{F}\|_{op} \cdot C(m,n)~,
\end{align}
where $C(m,n)=\sum_{\ell=n}^{m-1}[c^{+}_{\ell}]^{-1}=\sum_{\ell=n}^{m-1}[(J-\ell)(J+\ell+1)]^{-1/2}$ and we used the triangle inequality in the first inequality.

We can obtain a bound on $C(n,m)$ as  follows.
Consider the function $g(x):= [(J-x)(J+x+1)]^{-1/2}$ in the domain $x \in (-J-1,J)$. It is easy to see that the minimum of $g(x)$ is at $x^*=-1/2$, and $g(x)$ is a decreasing function from $x=-J-1$ to $x=x^*$ and an increasing function from $x=x^*$ to $x=J$. It is also easy to check that $g(x)=g(2x^* - x)$.
We first consider the case $n < 0 < m$.
We separate $C(n,m)$ into two pieces, $C(n,m)=C_1(n)+C_2(m)$, where $C_1(n):=\sum_{\ell=n}^{\ell=-1}[c^{+}_\ell]^{-1}$ and $C_2(m):=\sum_{\ell=0}^{m-1}[c^{+}_\ell]^{-1}$.
These two pieces can be bounded as 
\begin{align}
    \int_{n}^{0} g(x) dx  &< C_1(n) < \int_{n}^{0} g(x-1) dx ~, \notag\\ 
    \int_{0}^{m} g(x-1) dx  &< C_2(m) < \int_{0}^{m} g(x)dx ~.
\end{align}
The integrals of the upper bounds are given by $\int_0^{m}g(x)dx = 2\arcsin\left(\sqrt{\frac{J}{2J+1}}\right)-2\arcsin\left(\sqrt{\frac{J-m}{2J+1}}\right)$ and
$\int_{n}^{0}dx g(x-1)= \int_{n}^{0}dx g(-x)= \int_{0}^{-n}dy g(y)=2\arcsin\left(\sqrt{\frac{J}{2J+1}}\right)-2\arcsin\left(\sqrt{\frac{J+n}{2J+1}}\right)$.

We are interested in the regime where $n/J \rightarrow 0$,  $m/J \rightarrow 0$, and $J \sim N^a$ ($a \leq 1$) when $N \rightarrow \infty$. 
Consider the function 
\begin{equation}
    f(\tilde{x})= 2\arcsin\left( \sqrt{\frac{1+\tilde{x}}{2+J^{-1}}} \right)~,
\end{equation}
where $\tilde{x} = x/J$.
Its Taylor series is $f(\tilde{x})- f(0)=f'(0) \tilde{x} + O(f''(0)\tilde{x}^2)$, where $f(0)=2\arcsin\left(\sqrt{\frac{J}{2J+1}}\right)$, $f'(0)=\sqrt{\frac{J}{1+J}}$ and $f''(0)=-\frac{J^{\frac{1}{2}}}{2(1+J)^{\frac{3}{2}}}$.
Therefore, we have 
\begin{align}
    2\arcsin \left(\sqrt{\frac{J}{2J+1}} \right) - 2\arcsin \left(\sqrt{\frac{J+x}{2J+1}} \right) &= - \frac{x}{\sqrt{J(J+1)}} + O\left(\frac{x^2}{[J(1+J)]^{\frac{3}{2}}}\right) \notag \\
    &= -\frac{x}{J}(1+J^{-1})^{-\frac{1}{2}} + O\left(\frac{x^2}{J^3}(1+J^{-1})^{-\frac{3}{2}}\right) \notag \\
    &= -\frac{x}{J} + O(xJ^{-2})~.
\end{align}
Therefore, for $n < 0 < m$, we have 
\begin{equation}\label{eqn:Cmn}
    C(m,n) = \frac{|m-n|}{J} + O(|m-n|J^{-2})~.
\end{equation}
If $0 < n < m$, then we only need to bound $\sum_{\ell = n}^{m-1}[c^{+}_\ell]^{-1} < \int_{n}^m g(x)dx = \int_{0}^m g(x)dx - \int_{0}^n g(x)dx = |m-n|J^{-1}+ O(|m-n|J^{-2})$; if $ n < m < 0 $, we bound $\sum_{\ell=n}^{m-1}[c^{+}_\ell]^{-1} < \int_{n}^{m} g(x-1) dx = \int_{-m}^{-n} g(x) dx = |m-n|J^{-1}+ O(|m-n|J^{-2})$. 

Substituting the above results to Eq.~(\ref{eqn: Jdiff in C}), we have 
\begin{align}
    |\la J,{m}|\hat{F}|J,{m}\ra-\la J,n|\hat{F}|J,n\ra| \leq \|\hat{F}\|_{op}\cdot \left(\frac{dq_0|m-n|}{J} + O\left(\frac{d|m-n|}{J^2}\right)  \right)~.
\end{align}

Therefore, for $J\sim N^{a}$ and $|m-n| \sim N^{b}$, where $b < a \leq 1$, we have 
\begin{align}
    |\la J,{m}|\hat{F}|J,{m}\ra-\la J,n|\hat{F}|J,n\ra|&= \|\hat{F}\|_{op}\cdot [O(dN^{b-a}) + O(dN^{b-2a}) ] \notag \\
    &=\|\hat{F}\|_{op} \cdot O(dN^{b-a})~.
\end{align}

\end{proof}

\section{Proof of Theorem 1 and Corollary 1}\label{app:proof theorem 1}
In this section, we present the proof of Theorem 1, which gives an upper bound on the code inaccuracy against a generic $d$-local noise channel.

\textit{Theorem 1:  Consider the code space $\mathfrak{C}=\text{span}\{|J\!,\!M\ra; M\!=\!\Mmin, \Mmin\!+\!\Delta, \Mmin\!+\!2\Delta, \dots, \Mmax\}$, where $J\!\sim\! N^a $, $\Mmax \!-\!\Mmin \! \sim \! N^{b}$, and $b\! <\! a\! \leq 1$.  The code $\mathfrak{C}$ forms an $(\!(N,O[b\log_2 N])\!)$ AQECC against $d$-local noise on known sites, with inaccuracy $\epsilon(d) = O[\sqrt{d}(2s+1)^{d} N^{-(a-b)/2}]$ if $\Delta \geq 2sd+1$.}

\begin{proof} 
The proof of this theorem uses the Knill-Laflamme condition for the AQECC developed by B\'eny and Oreshkov in Ref.~\cite{benyGeneral2010}, which we now state.

\textit{(B\'eny and Oreshkov) Consider a noise channel with the Kraus representation $\mathcal{N}(\bullet)=\sum_{j}K_{j}(\bullet)K_j^{\dagger}$, and let $\Pi$ be the projector onto the code space. 
Then the code is $\epsilon$-correctable (i.e., there exists a recovery channel $\mathcal{R}$ such that $D_P(\mathcal{R}\mathcal{N},id) \leq \epsilon$) if and only if 
\begin{equation}\label{appeq: KL conditions}
\Pi K_i^{\dagger}K_j \Pi = \lambda_{ij}\Pi + \Pi B_{ij}\Pi~,    
\end{equation}
where $\lambda_{ij}$ are the components of a density operator, and $D_P(\Lambda+\mathcal{B},\Lambda)\leq \epsilon$, where $\Lambda(\rho)=\sum_{i,j}\lambda_{ij}\tr(\rho)|i\ra\la j|$ and $(\Lambda+\mathcal{B})(\rho)=\Lambda(\rho)+\sum_{i,j}\tr(\rho B_{ij})|i\ra\la j|$.
}

This theorem in fact has an interpretation from the complementary channel point of view, as described in Ref.~\cite{benyGeneral2010}. More specifically, consider the noise channel $\mathcal{N}(\rho)=\sum_{i}K_i (\rho)K_i^{\dagger}$, whose complementary channel is $\hat{\mathcal{N}}(\rho)=\sum_{i,j}\tr(K_j \rho K_i^{\dagger})|i\ra \la j|= (\Lambda+\mathcal{B})(\rho)$,  where $|i\ra \la j|$ acts on the environment Hilbert space. 
If we have perfect Knill-Laflamme conditions, the complementary channel of the noise becomes $\hat{\mathcal{N}}(\rho)=\sum_{i,j}\lambda_{ij}|i\ra \la j|:=\lambda_0 =\Lambda(\rho)$, which is a fixed density matrix $\lambda_0$ independent of the input $\rho$. In other words, the condition for having a perfect quantum error-correcting code is equivalent to the condition that the environment gains no information from the code space. 
And we indeed see that the AQECC generalization of the Knill-Laflamme conditions is a perturbation from the complementary channel point of view, where the environment can gain some $\epsilon$-small information from the code space.
Note that, in the original theorem in Ref.~\cite{benyGeneral2010}, Bures distance $D_B(\mathcal{N},\mathcal{M})$ is used as the measure of channel distance, while here we opt for using the purified distance $D_P(\mathcal{N},\mathcal{M})$, where $D_B(\mathcal{N},\mathcal{M})^2 \leq D_P(\mathcal{N},\mathcal{M})$ can be obtained using Eq.~(\ref{eqn:channel distance relations}).

Here we consider the noise channel $\mathcal{N}(\rho)=\sum_{i=1}^{n_d}K_i (\rho)K_i^{\dagger}$, where $n_d=(2s+1)^{2d}$ is the number of independent Kraus operators.
We assume $K_j$ is at most $d$-local, and that all $K_j$ act on the same set of $d$ qudits, so $K_i^{\dagger}K_j$ is at most $d$-local as well. 
We are also interested in the regime $J\sim N^{a}$ and $(\Mmax - \Mmin)  \sim N^{b}$.
From \textit{Lemma 1}, for $|J,n\ra, |J,m\ra  \in \mathfrak{C}$ and $n \neq m$, we have $\la J\!,\!{n}|K_i^{\dagger}K_j|J\!,\!{m}\ra = 0$ if $\Delta \geq 2sd+1$.
Now assuming we take $\lambda_{ij} = \la J\!,\!\Mmin | K^{\dagger}_i K_j |J\!,\!\Mmin\ra$,we then have $\Lambda(\rho)=\sum_{i,j=1}^{n_d}\lambda_{ij}\tr[\rho]|i\ra\la j|$ and $(\Lambda+\mathcal{B})(\rho)=\sum_{i,j=1}^{n_d}\tr[K_j \rho K_i^{\dagger}]|i\ra\la j|$, where the inputs of both channels are from the logical space.

Our goal now is to bound the entanglement fidelity,  $F(\Lambda,\Lambda+\mathcal{B}):=\min_{|\psi \ra}f((\Lambda+\mathcal{B}) \otimes \mathrm{id}(|\psi\ra \la \psi|), \Lambda \otimes\mathrm{id} (|\psi\ra \la \psi|) )$.
Recall that the inputs of $\Lambda$ and $\Lambda+\mathcal{B}$ are from the logical space, and the auxiliary space is isomorphic to the input space. Let us use $|x\ra := V^{\dagger}|J\!,\!M_x\ra$ for $x = 1, \cdots, 2^k$ to denote the basis of the logical space for convenience, where $|J\!,\!M_x\ra \in \mathfrak{C}$, and $V$ is the isometric encoding such that $VV^{\dagger}=\Pi$ is the projection onto the code space. 

For any state $|\psi\ra = \sum_{x,y}A_{xy}|x\ra\otimes |y\ra$, if we consider the operators $A=\sum_{x,y} A_{xy}|x\ra \la y|$ and $B:= A^{T}=\sum_{x,y} A_{xy}|y\ra \la x|$ and the state $|\Phi\ra=\sum_{x=1}^{2^k}|x\ra \otimes |x\ra$, then we can express $|\psi\ra = (A\otimes \mathrm{id}) |\Phi\ra = (\mathrm{id} \otimes B) |\Phi\ra$. 
We can therefore rewrite $\mathcal{E} \otimes \mathrm{id}[(V \otimes I)(|\psi\ra \la \psi|)(V^{\dagger} \otimes I)] =(\mathrm{id} \otimes B) (\mathcal{E} \otimes \mathrm{id})(|\tilde{\Phi}\ra \la \tilde{\Phi}|) (\mathrm{id} \otimes B^{\dagger})$ for $\mathcal{E} = \Lambda$ and $\mathcal{E}=\Lambda + \mathcal{B}$, where we abbreviate $|\tilde{\Phi}\ra := (V \otimes I)|\Phi\ra$. 
Moreover,
\begin{align}
    &(\Lambda+\mathcal{B}) \otimes \mathrm{id}(|\tilde{\Phi}\ra \la \tilde{\Phi}|) 
    = \sum_{x=1}^{2^k}\sum_{i,j=1}^{n_d} \la J\!,\!M_x|K_i^{\dagger}K_j|J\!,\!M_x \ra \cdot|i\ra \la j| \otimes |x\ra \la x| 
     :=  \sum_{x=1}^{2^k}\sum_{i,j=1}^{n_d} \Gamma^{x}_{ij}\cdot |i\ra \la j|\otimes |x\ra \la x|,\\
    &\Lambda \otimes \mathrm{id}(|\tilde{\Phi}\ra \la \tilde{\Phi}|) 
    = \sum_{x=1}^{2^k}\sum_{i,j=1}^{n_d} \la J\!,\!\Mmin|K_i^{\dagger}K_j|J\!,\!\Mmin \ra \cdot|i\ra \la j| \otimes |x\ra \la x|
    = \sum_{x=1}^{2^k}\sum_{i,j=1}^{n_d} \lambda_{ij} \cdot|i\ra \la j| \otimes |x\ra \la x|~,
\end{align}
where we have used Lemma~\ref{lemma 1} and denoted $\Gamma^{x}_{ij}=\la J\!,\!M_x|K_i^{\dagger}K_j|J\!,\!M_x\ra$.

From the normalization condition $\la \psi |\psi \ra = \sum_x \la x| B^{\dagger}B|x\ra =1$, we can interpret $\la x| B^{\dagger}B|x\ra $ as a probability.
The fidelity can be rewritten as 
\begin{align}\label{appeqn:fidelity join concavity rewrite}
    &f( (\Lambda+\mathcal{B}) \otimes \mathrm{id}(|\psi\ra \la \psi|), \Lambda \otimes\mathrm{id} (|\psi\ra \la \psi|) ) \notag \\
    &= f\left( \sum_{x=1}^{2^k}\la x|B^{\dagger}B| x\ra \cdot \sum_{i,j=1}^{n_d}\Gamma_{ij}^{x}|i\ra\la j|\otimes \frac{B |x \ra \la x| B^{\dagger}}{\la x|B^{\dagger}B| x\ra},  
    \sum_{x=1}^{2^k}\la x|B^{\dagger}B| x\ra \cdot\sum_{i,j=1}^{n_d}\lambda_{ij}|i\ra\la j|\otimes \frac{B |x \ra \la x| B^{\dagger}}{\la x|B^{\dagger}B| x\ra} \right) \notag \\
    & \geq \sum_{x=1}^{2^k}\la x|B^{\dagger}B| x\ra \cdot f\left(  \sum_{i,j=1}^{n_d}\Gamma_{ij}^{x}|i\ra\la j|\otimes \frac{B |x \ra \la x| B^{\dagger}}{\la x|B^{\dagger}B| x\ra},  
     \sum_{i,j=1}^{n_d}\lambda_{ij}|i\ra\la j|\otimes \frac{B |x \ra \la x| B^{\dagger}}{\la x|B^{\dagger}B| x\ra} \right) \notag \\
     & = \sum_{x=1}^{2^k}\la x|B^{\dagger}B| x\ra \cdot f\left(  \sum_{i,j=1}^{n_d}\Gamma_{ij}^{x}|i\ra\la j|,  
     \sum_{i,j=1}^{n_d}\lambda_{ij}|i\ra\la j| \right)~,
\end{align}
where we have used the joint concavity of the fidelity in the first inequality.

We will use the Fuchs-van de Graaf inequality Eq.~(\ref{eqn: fuchs-van de Graaf}) to bound the fidelity via the trace norm. 
To this end, we would like to bound $\|\Gamma_x - \lambda_0 \|_1$, where $\Gamma_x := \sum_{i,j=1}^{n_d}\Gamma_{ij}^{x}|i\ra\la j|$ and $\lambda_0 := \sum_{i,j=1}^{n_d}\lambda_{ij}|i\ra\la j|$.
Note that since $\Gamma_x^{\dagger} = \Gamma _x$ and $\lambda_0^{\dagger}=\lambda_0$, for each $x$, we can find a unitary transformation $U_x$ that diagonalizes $\Gamma_x - \lambda_0$, giving us  $\gamma_{i}^{x}\delta_{i\ell}=\sum_{j,k}[U^{\dagger}_x]_{ij}(\Gamma^x_{jk} - \lambda_{jk})[U_x]_{k\ell}$.
Furthermore,
\begin{align}
    \gamma_{i}^{x} &=\sum_{j,k}[U_x^{\dagger}]_{ij}(\la J\!,\!M|K_j^{\dagger}K_k|J\!,\!M\ra-\la J\!,\!\Mmin|K_j^{\dagger}K_k|J\!,\!\Mmin\ra)[U_x]_{ki} \notag \\
    &:= \la J\!,\!M|F_{i}^{\dagger}F_{i}|J\!,\!M\ra-\la J\!,\!\Mmin|F_{i}^{\dagger}F_{i}|J\!,\!\Mmin\ra~,
\end{align} 
where $F_{i} := \sum_{j}K_{j}[U_x]_{ji}$ is at most $d$-local and $\sum_{i}F_{i}^{\dagger}F_{i}=I$. (Note that $F_i$ depends on $x$ in general, though here we suppress it to simplify the notation.)
We therefore have $|\gamma_i^{x}| = \|F_{i}\|_{\text{op}}\cdot O(dN^{b-a}) = O(dN^{b-a})$ from Lemma~\ref{lemma 2} and the fact that $\|F_{i}\|_{op} \leq 1$.
This enables us to bound $\|\Gamma_x - \rho_0 \|_1 = \|\sum_{i=1}^{n_d} \gamma_i^x |i\ra \la i|\|_1 = O(dn_dN^{b-a})$, which is independent of $x$.
Using the Fuchs-van de Graaf inequality Eq.~(\ref{eqn: fuchs-van de Graaf}), or $f(\Gamma_x,\lambda_0) \geq 1 - \frac{1}{2}\|\Gamma_x-\lambda_0\|_1 = 1 - O(n_d dN^{b-a})$, we have 
\begin{align}
    f( (\Lambda+\mathcal{B}) \otimes \mathrm{id}(|\psi\ra \la \psi|), \Lambda \otimes\mathrm{id} (|\psi\ra \la \psi|) ) 
     & \geq \sum_{x=1}^{2^k}\la x|B^{\dagger}B| x\ra \cdot f\left(  \sum_{i,j=1}^{n_d}\Gamma_{ij}^{x}|i\ra\la j|,  
     \sum_{i,j=1}^{n_d}\lambda_{ij}|i\ra\la j| \right) \notag \\
     &\geq \sum_{x=1}^{2^k}\la x|B^{\dagger}B| x\ra \cdot [1 - O(dn_d N^{b-a})] \notag \\
     &= 1 - O(dn_d N^{b-a})~.
\end{align}
Since the above inequality is independent of $|\psi\ra$, we have $F(\Lambda,\Lambda+\mathcal{B})=1-O(dn_d N^{b-a})$.
We therefore have 
\begin{equation}
    \epsilon = D_P(\Lambda,\Lambda+\mathcal{B})=\sqrt{1-F(\Lambda,\Lambda+\mathcal{B})^2}=O(\sqrt{dn_d }N^{-(a-b)/2})=O(\sqrt{d}(2s+1)^d N^{-(a-b)/2})~,
\end{equation}
where we have used $n_d = (2s+1)^{2d}$.
Since $k=\log_2[\frac{\Mmax-\Mmin}{\Delta}+1]$, and $\Delta \geq 2sd+1$, we have $k = O(b\log_2 N)$ for $|\Mmax-\Mmin| \sim N^b$.

\end{proof}

We see that the reason to assume the Kraus operators $K_i$ all act on the same set of qudits is so that the number of Kraus operator $n_d=(2s+1)^{2d}$ can be controlled. 
In general, if we consider the noise channel $\mathcal{N}(\rho)=\sum_{i=1}^{n_d}K_i \rho K_i^{\dagger}$, where $K_i$ can act on different sets of at most $d$ qudits, then we have the following corollary.

\begin{corollary}\label{corollary 2}
Consider the code space $\mathfrak{C}=span\{|J\!,\!M\ra; M\!=\!\Mmin, \Mmin\!+\!\Delta, \Mmin\!+\!2\Delta, \dots, \Mmax\}$, where $J \sim N^a $, $ \Mmax \!-\!\Mmin \sim N^{b}$, and $b < a \leq 1$.  $\mathfrak{C}$ forms an $(\!(N,O(\log_2 \frac{N^b}{\Delta}))\!)$ AQECC against a $d$-local noise channel with $\epsilon(d) = O(\sqrt{d}n_d N^{(b-a)/2})$ if we take $\Delta \geq 4sd+1$.
\end{corollary}
\begin{proof}
    The proof follows from the proof of Theorem~\ref{theorem 1}. The difference is that, here, $K_i^{\dagger}K_j$ is at most $2d$-local instead of $d$-local, so we need to take $\Delta \geq 4sd+1$. 
    Moreover, since $F_{i}=\sum_{j}K_j [U_{x}]_{ji}$ is not necessary $d$-local any more, we can only bound 
    \begin{align}
        \|\Gamma_x - \lambda_0\|_1 & = O(d n_d^2 N^{-(a-b)})~,
    \end{align}
    since there are $n_d^2$ terms.
    This gives us
    \begin{equation}
    D_P(\Lambda,\Lambda+\mathcal{B})= O(\sqrt{d}n_d N^{-(a-b)/2})~.
\end{equation}
The above result illustrate the tradeoff between the parameters $d$, $n_d$, and the suppression factor $N^{-(a-b)/2}$

\end{proof}

\section{Code Performance Against i.i.d. errors}\label{app:proof Corollary 2}
While our Theorem 1 and Corollary 1 characterize the code’s performance against general errors, they assume knowledge of the error locations. 
This assumption effectively limits the number of Kraus operators.
In particular, the code has a “room for error” of order $O(N^{-(a-b)/2})$. 
This means the code cannot effectively correct errors while achieving an asymptotically vanishing inaccuracy if the number of Kraus operators grows faster than this “room for error”, as indicated by Corollary 1.
The assumption of knowing the locations of the error in Theorem 1 is therefore a way to restrict the number of Kraus operators of the error channel.
However, we emphasize that, in the proofs of both Theorem 1 and Corollary 1, we assumed that each Kraus operator has an operator norm of order one, which is a worst-case assumption. 
On the other hand, one often expects that error processes involving more qudits occur with a suppressed probability, and thus the corresponding Kraus operators would have smaller operator norms.
This consideration should allow for more Kraus operators, as long as the number of Kraus operators, weighted by their operator norms, grows more slowly than the “room for error.”

With this motivation, we analyze the code performance under the following i.i.d.~error model. 
For simplicity, we consider the noise channel $\mathcal{N} = \bigotimes_{j=1}^N \mathcal{N}_j$, where each local noise channel acts as
\begin{equation}
    \mathcal{N}_j(\rho) = (1-p)\rho + p A_j \rho A_j^{\dagger}~,
\end{equation}
with $A_j$ a unitary operator acting on site $j$, and $p$ the error probability.

This noise channel admits a Kraus representation
\begin{equation}
    \mathcal{N}(\rho) = \sum_{\vec{\ell}} E_{\vec{\ell}} \rho E_{\vec{\ell}}^{\dagger}~,
\end{equation}
where $\vec{\ell} = (\ell_1, \cdots, \ell_N) \in \{0,1\}^N$ is a binary string, and we use $\ell$ to denote its Hamming weight, defined as the number of entries equal to $1$ in $\vec{\ell}$.

The corresponding Kraus operator takes the form
\begin{equation}
    E_{\vec{\ell}} := \sqrt{p_{\ell}} E_{\ell_1} E_{\ell_2} \cdots E_{\ell_N}~,
\end{equation}
with $p_{\ell} = p^{\ell}(1-p)^{N-\ell}$, where $E_{\ell_j = 0} = I_j$ is the identity operator on site $j$, and $E_{\ell_j = 1} = A_j$ is the error operator on site $j$.

It is useful to decompose the channel into contributions from errors acting on at most $d$ sites, and those acting on more than $d$ sites:
\begin{equation}
    \mathcal{N} = P_d \mathcal{N}_d + P_{>d} \mathcal{N}_{>d}~,
\end{equation}
where $P_d = \sum_{\ell=0}^d \binom{N}{\ell} p^{\ell}(1-p)^{N-\ell}$ and $P_{>d} = 1 - P_d = \sum_{\ell=d+1}^N \binom{N}{\ell} p^{\ell}(1-p)^{N-\ell}$ are the respective probabilities. The maps $\mathcal{N}_d$ and $\mathcal{N}_{>d}$ are completely positive and trace-preserving (CPTP), describing the $d$-qudit errors and the remaining higher-weight errors, respectively.
More specifically, $\mathcal{N}_d(\rho) = \sum_{\vec{\ell} : \ell \leq d} K_{\vec{\ell}} \rho K_{\vec{\ell}}^{\dagger}$, where $K_{\vec{\ell}} = E_{\vec{\ell}} / \sqrt{P_d}$. For later convenience, we set $K_{\vec{\ell}} = 0$ if $\ell > d$.
As we will show, our code can protect against this noise model by correcting $d$-qudit errors $\mathcal{N}_d$, with $d \leq (c +o(1)) \log N$, provided the error probability $p$ satisfies $p \leq (c_0^2 + o(1))N^{-2} \log^2 N$. We now state this result.

\begin{corollary}\label{Corollary 2}
Consider the code space $\mathfrak{C}=span\{|J\!,\!M\ra; M\!=\!\Mmin, \Mmin\!+\!\Delta, \Mmin\!+\!2\Delta, \dots, \Mmax\}$, where $J \sim N^a $, $ \Mmax \!-\!\Mmin \sim N^{b}$, $b < a \leq 1$. Under the i.i.d.~noise channel 
$\mathcal{N}=\bigotimes_{j=1}^N \mathcal{N}_j$, where $\mathcal{N}_j(\rho)=(1-p)\rho + p A_j \rho A_j^{\dagger}$ for a unitary matrix $A_j$ on site $j$,
$\mathfrak{C}$ forms an $(\!(N,O(b\log_2 N))\!)$ AQECC by correcting $d$-qudit errors with $d = (c+o(1))\log N$, achieving a code inaccuracy $\epsilon = O(\sqrt{\log N}N^{-\beta/2})$, if $p = (c_0^2+o(1))(N^{-2}\log^2 N)$, $\Delta \geq 4sd + 1$, and $\beta = a-b-c_0-2c > 0$.
\end{corollary}

\begin{proof}
    We aim to lower bound the entanglement fidelity $\max_{\mathcal{R}}F(\mathcal{R}\mathcal{N},\id)$, in order to demonstrate that the purified distance is small.
    Note that, using the concavity of the fidelity, we have 
\begin{equation}
    f(\mathcal{R}\mathcal{N} \otimes \id (|\psi\ra \la \psi))\geq P_d f(\mathcal{R}\mathcal{N}_d \otimes \id (|\psi\ra \la \psi)) + P_{>d} f(\mathcal{R}\mathcal{N}_{>d} \otimes \id (|\psi\ra \la \psi)) 
\end{equation}
for any state $|\psi\ra$ and channel $\mathcal{R}$.
Suppose $|\psi^*\ra$ minimizes $f(\mathcal{R}\mathcal{N} \otimes \id (|\psi\ra \la \psi|))$. We then have $f(\mathcal{R}\mathcal{N}_d \otimes \id (|\psi^*\ra \la \psi^*|)) \geq \min_{\psi}f(\mathcal{R}\mathcal{N}_d \otimes \id (|\psi\ra \la \psi|))=F(\mathcal{R}\mathcal{N}_d,\id)$, and similarly $f(\mathcal{R}\mathcal{N}_{>d} \otimes \id (|\psi^*\ra \la \psi^*|)) \geq F(\mathcal{R}\mathcal{N}_{>d},\id)$.
We therefore have
\begin{equation}\label{eqn:strong_concavity_ent_fid}
    F(\mathcal{R}\mathcal{N},\id) \geq P_d F(\mathcal{R}\mathcal{N}_d,\id) + P_{>d} F(\mathcal{R}\mathcal{N}_{>d},\id)
\end{equation}
for any quantum channel $\mathcal{R}$.

Suppose we can lower bound $\max_{\mathcal{S}}F(\mathcal{S}\mathcal{N}_d,\id)\geq 1-\epsilon_d$, and assume $\mathcal{S}^*$ is the channel that minimizes the left-hand side. Then we have 
\begin{align}
    \max_{\mathcal{R}}F(\mathcal{R}\mathcal{N},\id)&\geq F(\mathcal{S}^{*}\mathcal{N},\id) \notag \\
    &\geq P_d F(\mathcal{S}^{*}\mathcal{N}_d,\id)+P_{>d}F(\mathcal{S}^{*}\mathcal{N}_{>d},\id) \notag \\
    & \geq P_d(1-\epsilon_d) = 1-(P_{>d}+P_d\epsilon_d)~,
\end{align}
where we have used the concavity of entanglement fidelity Eq.~(\ref{eqn:strong_concavity_ent_fid}), $F(\mathcal{S}^{*}\mathcal{N}_{d},\id) \geq (1-\epsilon_d)$, and $F(\mathcal{S}^{*}\mathcal{N}_{>d},\id) \geq 0$.
The purified distance is therefore 
\begin{equation}\label{appeqn: purified distance iid bound}
    \min_{\mathcal{R}}D_P(\mathcal{R}\mathcal{N},\id)\leq \sqrt{ 1-(1 - P_{>d}-P_d\epsilon_d)^2}  \leq \sqrt{2(P_{>d} + P_d\epsilon_d)}~.
\end{equation}
Physically, this inequality indicates that the code inaccuracy can be made small as long as the probability $P_{>d}$ is small, and the code achieves small inaccuracy $\epsilon_d$ when correcting $d$-local errors $\mathcal{N}_d$.
As we will show below, we can obtain an upper bound for $\min_{\mathcal{R}}D_P(\mathcal{R}\mathcal{N},\id)$ by bounding $P_{>d}$ and $P_d \epsilon_d$ from above.

First, we lower bound $\max_{\mathcal{S}}F(\mathcal{S}\mathcal{N}_d,\id)$. 
We follow the same strategy as the one used in the proofs of Theorem 1 and Corollary 1 by examining Knill-Laflamme conditions for the noise channel $\mathcal{N}_d(\rho) = \sum_{\vec{\ell}} K_{\vec{\ell}}\rho K_{\vec{\ell}}^{\dagger}$, where, for convenience, we set $K_{\vec{\ell}} =0$ if $\ell > d$.
Considering the code space $\mathfrak{C}=\text{span}\{|J\!,\!M\ra; M\!=\!\Mmin, \Mmin\!+\!\Delta, \Mmin\!+\!2\Delta, \dots, \Mmax\}$, from Lemma 1, we have $\la J\!,\!m| K_{\vec{r}}^{\dagger}K_{\vec{t}}|J\!,\!n\ra =0 $ if $|m - n|\geq 4sd+1$, which is satisfied if $\Delta \geq 4sd+1$.
For the diagonal Knill-Laflamme conditions, using Lemma 2, we have
\begin{equation}
    |\la J\!,\!m|K^{\dagger}_{\vec{r}}K_{\vec{t}}|J\!,\!m\ra - \la J\!,\!M_{\min}|K^{\dagger}_{\vec{r}}K_{\vec{t}}|J\!,\!M_{\min}\ra| = \frac{\sqrt{p_r p_t}}{P_d}\cdot O(d N^{-(a-b)})~,
\end{equation}
where we have used $\|K_{\vec{\ell}}\|_{\text{op}}=\sqrt{p_\ell}$, where $p_\ell=(1-p)^{N-\ell}p^{\ell}$ and $\ell$ is the Hamming weight of the binary vector $\vec{\ell}$.

Similar to the proof of Theorem 1, we take $\lambda_{\vec{r}\vec{t}}=\la J\!,\!M_{\min}|K^{\dagger}_{\vec{r}}K_{\vec{t}}|J\!,\!M_{\min}\ra$, and  the channels $\Lambda$ and $\hat{\mathcal{N}}_d=\Lambda+\mathcal{B}$ to be $\Lambda(\rho)=\sum_{\vec{r},\vec{t}}\lambda_{\vec{r}\vec{t}}\tr[\rho]|\vec{r}\ra\la\vec{t}|$ and $(\Lambda + \mathcal{B})(\rho)=\sum_{\vec{r},\vec{t}}\tr[K_{\vec{t}}\rho K^{\dagger}_{\vec{r}}] |\vec{r}\ra\la\vec{t}|$, respectively.
To bound the entanglement fidelity $F(\hat{\mathcal{N}}_d,\Lambda)$, consider the state $|\Phi\ra = \sum_{x=1}^{2^k}|x\ra \otimes |x\ra$. We then have
\begin{align}
    (\Lambda+\mathcal{B})\otimes \id (|\tilde{\Phi}\ra\la\tilde{\Phi}|) &= \sum_{x=1}^{2^k}\sum_{\vec{r},\vec{t}}\Gamma^x_{\vec{r}\vec{t}}\cdot |\vec{r}\ra \la \vec{t}| \otimes |x\ra\la x| \\
    \Lambda\otimes \id (|\tilde{\Phi}\ra\la\tilde{\Phi}|) &= \sum_{x=1}^{2^k}\sum_{\vec{r},\vec{t}}\lambda_{\vec{r}\vec{t}}\cdot |\vec{r}\ra \la \vec{t}| \otimes |x\ra\la x|~,
\end{align}
where $\Gamma^x_{\vec{r}\vec{t}}=\la J\!,\!M_x|K^{\dagger}_{\vec{r}}K_{\vec{t}}|J\!,\!M_x \ra$. 
Analogous to Eq.~(\ref{appeqn:fidelity join concavity rewrite}), we have 
\begin{align}
    f( (\Lambda+\mathcal{B}) \otimes \mathrm{id}(|\psi\ra \la \psi|), \Lambda \otimes\mathrm{id} (|\psi\ra \la \psi|) ) \geq \sum_{x=1}^{2^k}\la x|B^{\dagger}B| x\ra \cdot f\left(  \sum_{\vec{r},\vec{t}}\Gamma_{\vec{r}\vec{t}}^{x}|\vec{r}\ra\la \vec{t}|,  
     \sum_{\vec{r},\vec{t}}\lambda_{\vec{r}\vec{t}}|\vec{r}\ra\la \vec{t}| \right)
\end{align}
due to the joint concavity of fidelity.
We then proceed to bound the fidelity between the density matrices $\Gamma_x := \sum_{\vec{r},\vec{t}}\Gamma_{\vec{r}\vec{t}}^{x}|\vec{r}\ra\la \vec{t}|$ and $\lambda_0 := \sum_{\vec{r},\vec{t}}\lambda_{\vec{r}\vec{t}}|\vec{r}\ra\la \vec{t}|$ using the trace distance, obtaining
\begin{align}
    \|\Gamma_x - \lambda_0 \|_1 &\leq \sum_{\vec{r},\vec{t}}|\Gamma_{\vec{r}\vec{t}}^{x}-\lambda_{\vec{r}\vec{t}}| = \sum_{r=0}^d\sum_{t=0}^d\binom{N}{r}\binom{N}{t} \frac{1}{P_d}(1-p)^{N-\frac{1}{2}(r+t)}p^{\frac{1}{2}(r+t)}\cdot O(dN^{-(a-b)}) \notag \\
    & = P_d^{-1}(1-p)^N \mathbf{F}_d^2 \cdot O(dN^{-(a-b)})~,
\end{align}
where 
\begin{equation}
    \mathbf{F}_d := \sum_{\ell=0}^d \binom{N}{\ell} \left(\frac{p}{1-p}\right)^{\frac{\ell}{2}} =  \sum_{\ell=0}^d \binom{N}{\ell} p_1^{\ell} = \sum_{\ell=0}^d \binom{N}{\ell} p_2^{\ell} (1-p_2)^{N-\ell} \cdot (1+p_1)^N~,
\end{equation}
where $p_1 := \sqrt{p/(1-p)}$ and $p_2 := p_1/(1+p_1)$. 
We are interested in the regime $d = (c+o(1))\log N$. Since we assume $p =(c_0^2+o(1))N^{-2}\log^2 N$, we have $p_2 = \sqrt{p}/(\sqrt{1-p}+\sqrt{p})=(c_0 + o(1))N^{-1}\log N$, or  $Np_2  =(c_0 + o(1))\log N$.
Assuming $d=(1-\delta)Np_2$ for some $\delta \in (0,1)$, the Chernoff bound for the binomial distribution gives
\begin{equation}
    \sum_{\ell=0}^d \binom{N}{\ell} p_2^{\ell} (1-p_2)^{N-\ell} \leq \exp(-\frac{1}{2}\delta^2 Np_2) = O(N^{-\frac{1}{2}c_0\delta^2})~.
\end{equation}
Since $p_1 = p_2/(1-p_2)$, we have
\begin{equation}
    (1+p_1)^N = \exp(N\log(1+p_1)) = \exp(-N\log(1-p_2))=O(\exp(Np_2))=O(N^{c_0})~.
\end{equation}
Similarly, we have 
\begin{equation}
    (1-p)^N = \exp(N\log(1-p)) = O(\exp(-Np))=O(\exp(-c_0^2N^{-1}\log^2N))~,
\end{equation}
which is sub-leading compared to $N^{\alpha}$ for any exponent $\alpha$.
We therefore have 
\begin{equation}
    \|\Gamma_x - \lambda_0 \|_1 = P_d^{-1} \cdot O(d N^{-\alpha})~,
\end{equation}
where $\alpha = a-b + c_0(\delta^2 -2)=a-b-c_0-2c+c^2/c_0 + o(1)$, since $\delta = 1-d/(Np_2)=1-c/c_0+o(1)$.

Using Fuchs-van de Graaf inequality Eq.~(\ref{eqn: fuchs-van de Graaf}), we have $f(\Gamma_x,\lambda_0) \geq 1 - \frac{1}{2}\|\Gamma_x - \lambda_0\|$, and therefore
\begin{equation}\label{appeqn: Nd bound}
    \max_{\mathcal{S}}F(\mathcal{S}\mathcal{N}_d,\id) = \max_{\mathcal{S}'}F(\hat{\mathcal{N}}_d,\mathcal{S}'\circ\text{Tr})
    \geq  F(\hat{\mathcal{N}}_d,\Lambda) \geq 1 -P_d^{-1}\cdot O(d N^{-\alpha})~,
\end{equation}
where we used the result in Ref.~\cite{benyGeneral2010} for the first equality and the fact that $\text{Tr}[\bullet]$ is the complimentary channel of the identity channel $\id$.

Finally, to bound $P_{>d}=\sum_{\ell=d+1}^{N}\binom{N}{\ell}p^\ell(1-p)^{N-\ell}$, we have
\begin{equation}\label{appeqn: P>d bound}
    P_{>d} \leq \binom{N}{d+1}p^{d+1} < \left(\frac{Nep}{(d+1)}\right)^{d+1} <  \left(\frac{Nep}{d}\right)^{d+1} =  O(\exp(d\log \frac{Nep}{d}))=O(N^{-c \log N})~,
\end{equation}
which approaches zero faster than $N^{-\alpha}$ for any $\alpha >0$, and we have used $\binom{N}{m}<(\frac{Ne}{m})^{m}$.

Substituting the results from Eq.~(\ref{appeqn: Nd bound}) and Eq.~(\ref{appeqn: P>d bound}) into Eq.~(\ref{appeqn: purified distance iid bound}), we have 
\begin{equation}
    \epsilon = \min_{\mathcal{R}}D_P(\mathcal{R}\mathcal{N},\id) = O(\sqrt{\log N} N^{-\beta/2})~,
\end{equation}
where $\beta := a-b-c_0-2c < \alpha$.
We therefore obtain the result of Corollary 2 as long as $\beta = a-b-c_0-2c >0$.

\end{proof}

\section{Proof of Lemma 3}\label{app:proof lemma 3}
In this section, we present the proof of Lemma 3, which bounds the fidelity between the reduced density matrices from code words of different magnetic quantum numbers. This proof is in fact mostly parallel to the proof of Lemma 14 in Ref.~\cite{faistContinuous2020} but with a different and generalized physical interpretation, enabled by the usage of Clebsch-Gordan coefficients.

\textit{Lemma 3: Consider the irrep of SU(2) formed by $|J\!,\!M\ra$, where $J=sN$, and the reduced density matrix $\rho_M := \tr_{\bar{\alpha}}[|J\!,\!M\ra \la J\!,\!M|]$, where $\bar{\alpha}$ is the complement of the set of erased qudits $\alpha$. 
We have the fidelity $f(\rho_M,\rho_{M=0})= 1-O(\frac{d M^2}{sN^2})= 1-O(N^{2b+c-2})$ asymptotically for $M \sim N^b$, $d \sim N^c$, and $1 > b \geq c$. 
}

\begin{proof} 
Note that, since $|J\!,\!M\rangle = \sum_{m_1 =-j_1}^{j_1} C^{JM}_{j_1,m_1,j_2,m_2}|j_1\!,\!m_1\rangle \otimes |j_2\!,\!m_2\rangle$, where $C^{JM}_{j_1,m_1,j_2,m_2} = (\langle j_1,m_1| \otimes \langle j_2,m_2 |)|J\!,\!M\rangle $ is the Clebsch-Gordan coefficient, we have $\rho_M = \sum_{m_1 =-j_1}^{j_1} (C^{JM}_{j_1,m_1,j_2,M-m_1})^2 |j_1m_1\ra \la j_1 m_1 |$, where $j_1 = sd$ and $j_2=J-j_1$, since $J=sN$ is the irrep with the highest possible total spin quantum number with multiplicity one. 
The fidelity between $\rho_M$ and $\rho_{M=0}$ is therefore $f(\rho_M,\rho_{M=0})=\sum_{m_1 = -j_1}^{j_1} C^{JM}_{j_1,m_1,j_2,M-m_1} C^{JM=0}_{j_1,m_1,j_2,-m_1}$.

For $J=j_1 + j_2$ , the Clebsch-Gordan coefficient is
\begin{equation}
C^{JM}_{j_1,m_1,j_2,m_2}=\delta_{M=m_1+m_2}\left[ \frac{\binom{2j_1}{j_1+m_1}\binom{2J-2j_1}{J-j_1+M-m_1}}{\binom{2J}{J+M}}\right]^{\frac{1}{2}}~,
\end{equation}
where $\binom{N}{k}=\frac{N!}{k!(N-k)!}$ is the binomial coefficient.
The fidelity is therefore
\begin{equation}
    f(\rho_M,\rho_{M=0})=\left[\binom{2J}{J+M}\binom{2J}{J}\right]^{-\frac{1}{2}}\sum_{m_1=-j_1}^{j_1}\binom{2j_1}{j_1+m_1}\left[ \binom{2J-2j_1}{J-j_1+M-m_1} \binom{2J-2j_1}{J-j_1-m_1}\right]^{\frac{1}{2}}~.
\end{equation}
    
We are after the asymptotic behavior of the fidelity in the powers of $J$ (equivalently $N$), expecting $M \sim J^{b}$ and $j_1 \sim J^{c}$.
Using Stirling's approximation $\ln(N!)=N\ln N - N + \frac{1}{2}\ln(2\pi N) + \frac{1}{12N} + O(N^{-3})$, we have 
\begin{align}
    \ln \binom{2R}{R+\Delta}&=2R \ln (2R)-2R +\frac{1}{2}\ln(2\pi) + \frac{1}{2}\ln(2R)+\frac{1}{24R} + O(R^{-3}) \notag \\
    & - (R+\Delta) \ln (R+\Delta) + (R+\Delta) - \frac{1}{2}\ln(2\pi) - \frac{1}{2}\ln(R+\Delta) - \frac{1}{12(R+\Delta)} + O((R+\Delta)^{-3}) \notag \\
    & - (R-\Delta) \ln (R-\Delta) + (R-\Delta) - \frac{1}{2}\ln(2\pi) - \frac{1}{2}\ln(R-\Delta) - \frac{1}{12(R-\Delta)} + O((R-\Delta)^{-3})~.
\end{align}
Expecting $\Delta /R \rightarrow 0$ when $R \rightarrow \infty$, but also the possibility that $R \sim N$ and $\Delta \sim N^{b}$ ($b<1$), we use 
\begin{align}
    \ln (R+\Delta) &= \ln R + \ln(1+\frac{\Delta}{R}) = \ln R +\frac{\Delta}{R} - \frac{\Delta^2}{2R^2} + \frac{\Delta^3}{3R^3}+O(\Delta^4R^{-4}) \\
    \ln (R-\Delta) &= \ln R + \ln(1-\frac{\Delta}{R}) = \ln R -\frac{\Delta}{R} - \frac{\Delta^2}{2R^2} - \frac{\Delta^3}{3R^3}+O(\Delta^4R^{-4}) \\
    (R \pm \Delta)^{-1} &=\frac{1}{R} \mp \frac{\Delta}{R^2} + O(\Delta^2R^{-3})~,
\end{align}
and obtain 
\begin{align}\label{eqn:binom_asym}
    \ln \binom{2R}{R+\Delta}&= (2\ln 2) R -\frac{1}{2}\ln R -\frac{1}{2}\ln \pi -\frac{1}{8R} -\frac{\Delta^2}{R} + \frac{\Delta^2}{2R^2} + O(\Delta^4R^{-3})~,
\end{align}
or 
\begin{align}\label{eqn:binom_asym_exp}
    \binom{2R}{R+\Delta}&= \frac{2^{2R}}{\sqrt{\pi R}}\exp\left ( -\frac{1}{8R} - \frac{\Delta^2}{R} + \frac{\Delta^2}{2R^2} +O(\Delta^4R^{-3}) \right)~.
\end{align}
Using Eq.~(\ref{eqn:binom_asym_exp}), we have 
\begin{align}
    \left[\binom{2J}{J+M}\binom{2J}{J}\right]^{-\frac{1}{2}}=\frac{2^{2J}}{\sqrt{\pi J}} \exp\left( -\frac{1}{8J}-\frac{M^2}{J}+\frac{M^2}{J^2} + O(M^4J^{-3})   \right)~.
\end{align}

Applying Eq.(\ref{eqn:binom_asym}) to $\binom{2J-2j_1}{J-j_1+M-m_1}$ and $\binom{2J-2j_1}{J-j_1-m_1}$, we have 
\begin{align}
    \ln \binom{2J-2j_1}{J-j_1+M-m_1} &= (2\ln 2)(J-j_1)-\frac{1}{2}\ln(J-j_1)+\frac{1}{2}\ln \pi-\frac{1}{8(J-j_1)}-\frac{(M-m_1)^2}{(J-j_1)}+\frac{(M-m_1)^2}{2(J-j_1)^2} + O(J^{-3}), \\
    \ln \binom{2J-2j_1}{J-j_1-m_1} &= (2\ln 2)(J-j_1)-\frac{1}{2}\ln(J-j_1)+\frac{1}{2}\ln \pi-\frac{1}{8(J-j_1)}-\frac{m_1^2}{(J-j_1)}+\frac{m_1^2}{2(J-j_1)^2} + O(J^{-3})~.
\end{align}

Using 
\begin{align}
    (J-j_1)^{-1}&=J^{-1}\left(1+\frac{j_1}{J}+O(J^{-2}) \right)~, \notag \\
    (J-j_1)^{-2}&=J^{-2}+O(J^{-3})~,
\end{align}
we have
\begin{align}
    \left[ \binom{2J-2j_1}{J-j_1+M-m_1} \binom{2J-2j_1}{J-j_1-m_1}\right]^{\frac{1}{2}} &= \frac{2^{2(J-j_1)}}{\sqrt{\pi(J-j_1)}}\exp \left( -\frac{1}{8J}-\frac{j_1}{8J^2} - \frac{(M-m_1)^2+m_1^2}{2J} \right. \notag \\
    &\left. + \frac{(M-m_1)^2+m_1^2}{4J^2}(1-2j_1) + O(J^{-3}) \vphantom{\int_1^2} \right)~.
\end{align}

Combining the above equations together, using $(1-j_1/J)^{-1/2}=(1+j_1/2J+3j_1^2/8J^2+O(J^{-3}))$, and expanding the equation up to $O(J^{-3})$, we have 
\begin{align}
    f(\rho_M,\rho_{M=0})&=\frac{2^{-2j_1}}{\sqrt{1-j_1/J}}\exp\left(-\frac{j_1}{8J^2}-\frac{j_1M^2}{2J^2}+O(J^{-3}) \right) \notag \\
    &\times \sum_{m_1=-j_1}^{j_1} \binom{2j_1}{j_1+m_1}\exp\left( \frac{Mm_1-m_1^2}{J}-\frac{Mm_1-m_1^2}{2J^2}(1-2j_1) \right) \notag \\
    &= 2^{-2j_1}\left(1+\frac{j_1}{2J}+\frac{3j_1^2}{8J^2} +O(J^{-3}) \right)\left(1-\frac{j_1(1+4M^2)}{8J^2}+O(J^{-3}) \right) \notag \\
    &\times \sum_{m_1=-j_1}^{j_1} \binom{2j_1}{j_1+m_1}\left( 1+\frac{Mm_1-m_1^2}{J}+\frac{(Mm_1-m_1^2)^2}{2J^2}-\frac{(Mm_1-m_1^2)}{2J^2}(1-2j_1)+O(J^{-3})\right) \notag \\
    &=2^{-2j_1} \sum_{m_1=-j_1}^{j_1} \binom{2j_1}{j_1+m_1}(1+J^{-1}A_{M,j_1,m_1}+J^{-2}B_{M,j_1,m_1})~,
\end{align}
where
\begin{align}
    A_{M,j_1,m_1} &= \frac{1}{2}j_1 + Mm_1 - m_1^2~,\notag \\
    B_{M,j_1,m_1} &= \frac{3j_1^2}{8}-\frac{j_1}{8}-\frac{1}{2}M^2j_1+\frac{1}{2}j_1(Mm_1-m_1^2)+\frac{1}{2}(Mm_1-m_1)^2+(j_1-\frac{1}{2})(Mm_1-m_1^2)~.
\end{align}

Upon the change of variables $r=j_1+m_1$, using equations 
\begin{align}
\sum_{r=0}^{2j_1}\binom{2j_1}{r} &= 2^{2j_1}~, \notag \\
\sum_{r=0}^{2j_1}\binom{2j_1}{r} r &= j_1 2^{2j_1}~, \notag \\
\sum_{r=0}^{2j_1}\binom{2j_1}{r} r^2 &= (\frac{1}{2}j_1+j_1^2) 2^{2j_1}~,
\end{align}
we have $\sum_{r=0}^{2j_1} \binom{2j_1}{r} A_{M,j_1,m_1}=0$.

In addition, using equations  
\begin{align}
\sum_{r=0}^{2j_1}\binom{2j_1}{r} r^3 &= \frac{1}{2}j_1^2(2j_1+3) 2^{2j_1}~, \notag \\
\sum_{r=0}^{2j_1}\binom{2j_1}{r} r^4 &= \frac{1}{4}j_1(1-3j_1+16j_1^2+4j_1^3) 2^{2j_1}~,
\end{align}
we have 
\begin{align}
    2^{-2j_1}\sum_{r=0}^{2j_1} \binom{2j_1}{r} B_{M,j_1,m_1}=-\frac{1}{4}j_1M^2-\frac{7}{8}j_1^2 +2j_1^2M^2 +4j_1^3M +\frac{1}{2}j_1^3+\frac{3j_1}{8}~.
\end{align}
So we have 
\begin{align}
    f(\rho_M,\rho_{M=0})&=1-J^{-2}\left( \frac{1}{4}j_1M^2+\frac{7}{8}j_1^2 -2j_1^2M^2 -4j_1^3M -\frac{1}{2}j_1^3-\frac{3j_1}{8}  \right) + O(J^{-3})  \notag \\
    &= 1-O\left(\frac{j_1M^2}{J^2}\right) = 1-O\left(\frac{dM^2}{sN^2}\right)= 1-O(N^{2b+c-2})~,
\end{align}
where we have used $J=sN$, $M \sim J^{b}$ and $j_1 = sd \sim J^{c}$ with $1 > b \geq c$. 
\end{proof}

\section{Proof of Theorem 2}\label{app:proof theorem 2}
In this section, we present the proof of Theorem 2, which bounds the code inaccuracy against the $d$-local erasure error. This theorem in fact can be obtained by combining our Lemma 1, our Lemma 3, and Theorem 3 of Ref.~\cite{faistContinuous2020}. Here we present the proof tailored to our case to make our presentation self-contained.

\textit{
    Theorem 2: Consider the code space $\mathfrak{C}=\text{span}\{|J\!,\!M\ra; M\!=\!\Mmin, \Mmin\!+\!\Delta, \Mmin\!+\!2\Delta, \dots, \Mmax\}$, where $J=sN$, $\Mmax-\Mmin \sim N^{b}$ and $b < 1$.  $\mathfrak{C}$ forms an $(\!(N,O[(b\!-\!c)\log_2N])\!)$ AQECC against heralded $d$-local erasures of size $d \sim N^{c}$, where $0 \leq c \leq b < 1$, with $\epsilon(d)=O(N^{-(1-b-c/2)})$ if we take $\Delta \geq 2sd +1$.
}

\begin{proof}
    The bound of the code inaccuracy against erasure errors can be obtained using Lemma~\ref{lemma 3} and Theorem~3 in Ref.~\cite{faistContinuous2020}. (See also Ref.~\cite{liuApproximate2023}.) 
    A potential concern is the applicability of the aforementioned theorem in Ref.~\cite{faistContinuous2020} when $d \sim N^c$ for some power $c < 1$, as the theorem in Ref.~\cite{faistContinuous2020} is only addressing the case with $d=O(1)$.  
    To alleviate the above concern and to make our presentation self-contained, we prove our code-inaccuracy bound against erasure errors at known locations.
    
    We consider the erasure errors at known locations expressed as $\mathcal{N}(\bullet)=\sum_{\alpha}p_{\alpha} |\alpha\ra \la \alpha|_{C} \otimes \mathcal{N}^{\alpha}(\bullet)$, where $\alpha$ denotes the set of qudits that are erased whose cardinal number is at most $d$, $C$ is the classical registry that records the set of lost qudits, and ${\cal N}^{\alpha}(\bullet):= \tr_{\alpha}(\bullet)$.
    To bound the code inaccuracy, following B\'eny and Oreshkov~\cite{benyGeneral2010}, we bound the entanglement fidelity between the complementary error channel $\widehat{N}(\rho) = (\Lambda+\mathcal{B})(\rho)$ and some channel $\Lambda(\rho)$, which needs to be independent of the input $\rho$.
    Note that the complementary channel of $\mathcal{N}^{\alpha}$ is $\widehat{\mathcal{N}}^{\alpha}(\bullet)=\tr_{\bar{\alpha}}(\bullet)$, where $\bar{\alpha}$ is the complement of $\alpha$.
    The complementary channel \(\widehat{\mathcal{N}}\) to the full erasure channel is just the original channel, but with each \(\mathcal{N}^\alpha\to\widehat{\mathcal{N}}^\alpha\) and with a different reference system \(C^\prime\).
    We consider the channels $\Lambda(\rho) := \sum_{\alpha}p_{\alpha}|\alpha\ra \la\alpha|_{C'} \otimes \mathcal{T}^{\alpha}(\rho)$ and $\mathcal{T}^{\alpha}(\rho) := \tr[\rho] \cdot \tau^{0}_{\alpha}$, both independent of the input $\rho$ and $\tau^{0}_{\alpha}:=\tr_{\bar{\alpha}}(|J,M\!=\!0 \rangle \langle J,M\!=\!0 |)$.
   The goal now is to bound the entanglement fidelity $F(\mathcal{\hat{N}}=\Lambda+\mathcal{B},\Lambda)$.
    
    Recall that the code space is $\mathfrak{C}=span\{|J\!,\!M\ra; M\!=\!\Mmin, \Mmin\!+\!\Delta, \Mmin\!+\!2\Delta, \dots, \Mmax\}$, and we use $|x\ra_L = V^{\dagger}|J,M_x\ra_A$, $x=1, \cdots, 2^k$,  to denote the basis in the logical space for $|J,M_x\ra \in \mathfrak{C}$, where $V$ is the isometric encoding such that $VV^{\dagger}=\Pi$ is the projection onto the code space. 
    Consider $\rho_\alpha^{x,x'}=\tr_{\bar{\alpha}}[|J,M_x\ra \la J,M_{x'}|]$, which is an at-most $d$-local operator. If we pick $\Delta \geq 2sd+1$, from Lemma~\ref{lemma 1}, we should have $\rho_\alpha^{x,x'} = 0$ if $ x \neq x'$, assuming the size of $\alpha$ is at most $d$. We therefore will only consider the case $x=x'$ and denote $\rho^{x}_{\alpha} := \tr_{\bar{\alpha}}[|J,M_x\ra \la J,M_{x}|]$ for later convenience.

First, using the joint concavity of fidelity for the probability $p_{\alpha}$, we have 
\begin{align}
    f(\widehat{\mathcal{N}}\otimes \mathrm{id}(|\psi\ra \la\psi|), \Lambda \otimes \mathrm{id} (|\psi\ra \la \psi |))& =f\left(\sum_{\alpha}p_{\alpha}|\alpha\ra\la \alpha|_C' \otimes (\widehat{\mathcal{N}}^{\alpha}\otimes \mathrm{id})(|\psi\ra \la\psi|), \sum_{\alpha}p_{\alpha}|\alpha\ra\la \alpha|_C' \otimes ({\mathcal{T}}^{\alpha}\otimes \mathrm{id})(|\psi\ra \la\psi|) \right) \notag \\
    &\geq \sum_{\alpha}p_{\alpha} f[(\widehat{\mathcal{N}}^{\alpha}\otimes \mathrm{id})(|\psi\ra \la\psi|),({\mathcal{T}}^{\alpha}\otimes \mathrm{id})(|\psi\ra \la\psi|)]~.
\end{align}

To bound $f[(\widehat{\mathcal{N}}^{\alpha} \otimes \mathrm{id})(|\psi\ra \la\psi|),({\mathcal{T}}^{\alpha} \otimes \mathrm{id})(|\psi\ra \la\psi|)]$, recall that the auxiliary space of $[\widehat{\mathcal{N}}^{\alpha}\otimes \mathrm{id}]$ and $[{\mathcal{T}}^{\alpha} \otimes \mathrm{id}]$ is isomorphic to the input space, which is the logical space.
For any state $|\psi\ra = \sum_{x,y}A_{xy}|x\ra\otimes |y\ra$, 
consider the operators $A=\sum_{x,y} A_{xy}|x\ra \la y|$ and $B:= A^T =\sum_{x,y} A_{xy}|y\ra \la x|$ and the state $|\Phi\ra= \sum_{x=1}^{2^k}|x\ra\otimes |x\ra$, we have $|\psi\ra = A \otimes \mathrm{id} |\Phi\ra = \mathrm{id} \otimes B |\Phi\ra$. 
Therefore, we can write
\begin{align}
    (\widehat{\mathcal{N}}^{\alpha} \otimes \mathrm{id})(|\psi\ra \la\psi|)&= (\mathrm{id}\otimes B) \widehat{\mathcal{N}}^{\alpha}(|\Phi\ra \la \Phi|)(\mathrm{id}\otimes B^\dagger) = \sum_{x=1}^{2^k} \rho_{\alpha}^x \otimes B|x\ra \la x|B^\dagger  \notag \\
    (\mathcal{T}^{\alpha} \otimes \mathrm{id})(|\psi\ra \la\psi|)&= (\mathrm{id}\otimes B) \widehat{\mathcal{N}}^{\alpha}(|\Phi\ra \la \Phi|)(\mathrm{id}\otimes B^\dagger) = \sum_{x=1}^{2^k} \tau_{\alpha}^0 \otimes B|x\ra \la x|B^\dagger~.
\end{align}

Recall that the normalization condition $\la \psi | \psi \ra = \sum_x \la x |B^{\dagger} B | x \ra =1$, so we can interpret $\la x| B^{\dagger} B | x \ra$ as a probability.
Using the joint concavity of fidelity again, we have 
\begin{align}
    &f[(\widehat{\mathcal{N}}^{\alpha} \otimes \mathrm{id})(|\psi\ra \la\psi|),({\mathcal{T}}^{\alpha} \otimes \mathrm{id})(|\psi\ra \la\psi|)]  \notag \\
    &= f\left( \sum_{x=1}^{2^k}\la x|B^{\dagger}B| x\ra \cdot \rho^x_\alpha \otimes \frac{B |x \ra \la x| B^{\dagger}}{\la x|B^{\dagger}B| x\ra},  
    \sum_{x=1}^{2^k}\la x|B^{\dagger}B| x\ra \cdot \tau^0_\alpha \otimes \frac{B |x \ra \la x| B^{\dagger}}{\la x|B^{\dagger}B| x\ra} \right) \notag \\
    & \geq \sum_{x=1}^{2^k}\la x|B^{\dagger}B| x\ra \cdot f\left(\rho^x_\alpha \otimes \frac{B |x \ra \la x| B^{\dagger}}{\la x|B^{\dagger}B| x\ra},  
     \tau^0_\alpha \otimes \frac{B |x \ra \la x| B^{\dagger}}{\la x|B^{\dagger}B| x\ra} \right) \notag \\
     & = \sum_{x=1}^{2^k}\la x|B^{\dagger}B| x\ra \cdot f(\rho^x_\alpha ,  
     \tau^0_\alpha) \notag \\
    &\geq \sum_x \la x|B^{\dagger} B|x\ra\left(1-O\left(\frac{j_1 \Mmax^2}{4J^2}\right) \right) \notag \\
    & = 1-O\left(\frac{j_1 \Mmax^2}{4J^2}\right) = 1-O\left(\frac{|\alpha| \Mmax^2}{4sN^2}\right) \notag \\ 
    & = 1-O\left(\frac{d \Mmax^2}{4sN^2}\right) ~,
\end{align}
where we have used Lemma~\ref{lemma 3}, $j_1=s|\alpha|$, $J=sN$, $|\alpha| \leq d$, and assumed $|\Mmax| \geq |\Mmin|$, without loss of generality.
We therefore have 
\begin{equation}
    f(\widehat{\mathcal{N}}\otimes \mathrm{id}(|\psi\ra \la\psi|), \Lambda \otimes \mathrm{id} (|\psi\ra \la \psi |)) = 1 - O\left(\frac{d \Mmax^2}{4sN^2}\right)~.
\end{equation}
Since it holds for arbitrary $|\psi\ra$, we have $F(\widehat{N},\Lambda) = 1- O\left(\frac{d \Mmax^2}{4sN^2}\right)$, or $\epsilon = D_P(\mathcal{\hat{N}},\Lambda)= O\left(\sqrt{\frac{d}{2s}}\frac{\Mmax}{N} \right)$. 
In the scaling regime $\Mmax-\Mmin \sim N^{b}$ and $d \sim N^{c}$, we have $\epsilon = O(N^{b+c/2-1})$.
The number of logical qubits is $k=\log_2\left( \frac{\Mmax-\Mmin}{\Delta} +1\right)$. Since $d \sim N^c$, we need to take the spacing of the magnetic quantum number $\Delta = \Omega(N^{c})$. We therefore have $k = O((b-c)\log_2 N)$ as claimed.

\end{proof}

\section{Calculation of QFI loss and measurements}\label{app:QFIloss}
In this section, we present the details of the numerical calculation for the QFI loss shown in the main text. In addition, we propose a simple $\theta$-independent local measurement for the noiseless probe state and a $\theta$-independent global measurement for the partially-erased probe state that can saturate the quantum Cram\'er-Rao bound. We also propose a local measurement for the partially-erased probe state, which can give a scaling better than the standard quantum limit in some parameter regime, though it does not saturate the quantum Cram\'er-Rao bound. 

We consider a system with $N$ spin-$s$ degrees of freedom and the ideal probe state $|\psi\ra =\frac{1}{\sqrt{2}}(|J,M\ra + |J,-M\ra)$, where $J=sN$, and the quantum metrology problem where the sensing parameter $\theta$ couples to $\hat{Q}_z$, $|\psi_{\theta}\ra = e^{-i\theta\hat{Q}_z}|\psi\ra$, or $\psi_{\theta}=\frac{1}{2}(|J,M\ra\la J,M|+|J,-M\ra\la J,-M|+e^{-i2M\theta}|J,M\ra\la J,-M|+e^{i2M\theta}|J,-M\ra\la J,M|)$.
Since $\hat{Q}_z|J,M\ra=M|J,M\ra$, we see that the QFI of this probe state is $\mathfrak{F}(\psi_{\theta})=4(\la \hat{Q}_z^2\ra-\la \hat{Q}_z\ra^2)=4M^2$.

Consider the probe state going through an erasure channel $\rho_{\theta}=\mathcal{N}(\psi_{\theta}):=\tr_{d}[\psi_{\theta}]$, where $\tr_{d}[...]$ means tracing out $d$ sites. Note that, since the $J=sN$ irrep is totally permutationally symmetric, we can just assume that we trace out the first $d$ sites.
To calculate the QFI of $\rho_{\theta}$, we need to find its eigendecomposition. This can be achieved as follows. 
First, note that we can rewrite 
\begin{equation}
    |J,M\ra = \sum_{m_1=-j_1}^{j_1}C^{J,M}_{j_1,m_1,j_2,M-m_1}|j_1,m_1\ra_d \otimes |j_2,M-m_1\ra_{\bar{d}}~,
\end{equation}
where $C^{J,M}_{j_1,m_1,j_2,m_2}$ are the Clebsch-Gordan coefficients, $j_1 = sd$, $j_2 = J-j_1$, and $\bar{d}$ is the complement of the $d$ sites to be traced out.
Defining $a_{m}=C^{J,M}_{j_1,m,j_2,M-m}$ and $b_{m}=C^{J,-M}_{j_1,m,j_2,-M-m}$, we have 
\begin{align}
    \mathcal{N}(|J,M\ra \la J,M|)&=\sum_{m=-j_1}^{j_1}a_m^2|j_2,M-m\ra\la j_2,M-m|, \notag \\
    \mathcal{N}(|J,-M\ra \la J,-M|)&=\sum_{m=-j_1}^{j_1}b_m^2|j_2,-M-m\ra\la j_2,-M-m|, \notag \\
    \mathcal{N}(|J,M\ra \la J,-M|)&=\sum_{m=-j_1}^{j_1}a_mb_m|j_2,M-m\ra\la j_2,-M-m|, \notag \\
    \mathcal{N}(|J,-M\ra \la J,M|)&=\sum_{m=-j_1}^{j_1}a_mb_m|j_2,-M-m\ra\la j_2,M-m|~.
\end{align}
Note that we consider the case $M > j_1$ so that $|j_2,M-m\ra$ and $|j_2,-M-m'\ra$ are always distinct states for $m,m' = -j_1 \dots j_1$. 
Therefore, we obtain 
\begin{align}
    \rho_{\theta}=\frac{1}{2}\bigoplus_{m=-j_1}^{j_1} \begin{pmatrix} 
a_m^2 & e^{-i2M\theta}a_mb_m \\
e^{i2M\theta}a_mb_m & b_m^2 
\end{pmatrix}:=  \bigoplus_{m=-j_1}^{j_1} B_m~,
\end{align}
where the matrix $B_m$ is represented by the basis $|j_2,M-m\ra$ and $|j_2,-M-m\ra$.
We also see that $\rho_{\theta}$ can be interpreted as an ensemble labeled by $m = -j_1, \cdots, j_1$, where the $\theta$ information is encoded in the coherence between the states $|j_2,M\!-\!m\ra$ and $|j_2,-M\!-\!m\ra$, which have a magnetic quantum number separation $2M$.
The eigenvalues of the $2\times 2$ matrix $B_m$ are $\frac{1}{2}(a_m^2+b_m^2)$ and $0$, and the corresponding eigenvectors are
\begin{align}
    |m\ra &= \frac{1}{\sqrt{2\lambda_m}}(e^{-iM\theta}a_m|j_2,M\!-\!m\ra + e^{iM\theta}b_m |j_2,-M\!-\!m\ra)~, \\
    |\bar{m}\ra &= \frac{1}{\sqrt{2\lambda_m}}(e^{-iM\theta}b_m|j_2,M\!-\!m\ra - e^{iM\theta}a_m |j_2,-M\!-\!m\ra)~,
\end{align}
respectively.
Note that $\partial_{\theta}\rho_{\theta}=-iM\sum_{m=-j_1}^{j_1}a_m b_m (e^{-2iM\theta}|j_2,M\!-\!m\ra \la j_2,-\!M\!-\!m| - e^{2iM\theta}|j_2,-\!M\!-\!m\ra \la j_2,M\!-\!m|)$.
Therefore, we obtain $\la m|\partial_{\theta}\rho_{\theta}|m'\ra = 0 $ for both $m \neq m'$ and $m=m'$.
On the other hand, $(\partial_{\theta}\rho_{\theta})^2=\oplus_{m=-j_1}^{j_1} M^2 a_m^2 b_m^2 I_{m}$, where $I_m$ is the identity matrix in the $|J,M\!-\!m\ra,|J,-M\!-\!m\ra$ subspace. 
Denoting the set of $m$ such that $\lambda_m \neq 0$ as $\Lambda$ (it is possible that some of $\lambda_m$ are zero), we have 
\begin{align}\label{eqn:QFIloss from CG coefficient}
    \mathfrak{F}(\rho_{\theta};\partial_{\theta}\rho_{\theta})&=\sum_{m,m':\lambda_m+\lambda_{m'}>0}\frac{2|\la m|\partial_{\theta}\rho_{\theta}|m'\ra|^2}{\lambda_m+\lambda_{m'}} \notag \\
    &=\sum_{m,m'\in\Lambda} \frac{2|\la m|\partial_{\theta}\rho_{\theta}|m'\ra|^2}{\lambda_m+\lambda_{m'}} +\sum_{m \in \Lambda} \frac{4}{\lambda_m}\la m| \partial_{\theta}\rho_{\theta}(I-\sum_{m'\in \Lambda}|m'\ra\la m'|) \partial_{\theta}\rho_{\theta} |m\ra \notag \\
    &=\sum_{m\in\Lambda} \frac{1}{\lambda_m}|\la m|\partial_{\theta}\rho_{\theta}|m\ra|^2 +\sum_{m \in \Lambda} \frac{4}{\lambda_m}\la m| \partial_{\theta}\rho_{\theta}(I-\sum_{m'\in \Lambda}|m'\ra\la m'|) \partial_{\theta}\rho_{\theta} |m\ra \notag \\
    &=\sum_{m \in \Lambda} \frac{4a_m^2 b_m^2}{\lambda_m} M^2 = \sum_{m \in \Lambda} \left(\frac{2a_m^2b_m^2}{a_m^2+b_m^2}\right) 4M^2~.
\end{align}
Finally, we note that $\mathfrak{F}(\psi_{\theta};\partial_{\theta}\psi_{\theta})=4M^2$ and $\Delta \mathfrak{F}=\mathfrak{F}(\psi_{\theta};\partial_{\theta}\psi_{\theta})-\mathfrak{F}(\rho_{\theta};\partial_{\theta}\rho_{\theta})$, and the Clebsch-Gordan coefficients can be calculated efficiently via a recursive formula or an explicit formula, which gives us an efficient method to calculate the QFI loss.
We calculate the QFI by picking values of $J$ and $j_1$. So the result is applicable to different $s$ with the corresponding $N=J/s$ and $d=j_1/s$.

It is interesting to examine the QFI loss if the probe state is a GHZ(-like) state $|\psi'\ra = \frac{1}{\sqrt{2}}(|J,J\ra + |J,-J\ra)$, which would achieve the Heisenberg scaling in the absence of noise.
We can express $|J,J\ra = |j_1,j_1\ra \otimes |j_2,j_2\ra$ and $|J,-J\ra = |j_1,-j_1\ra \otimes |j_2,-j_2\ra$, which gives us $\rho_{\theta}=\frac{1}{2}(|j_2,j_2\ra \la j_2,j_2|+|j_2,-j_2\ra \la j_2,-j_2|)$ independent of $\theta$. This implies that $\mathfrak{F}(\rho_{\theta};\partial_{\theta}\rho_{\theta}) = 0$.
This can also be obtained from Eq.~(\ref{eqn:QFIloss from CG coefficient}). Note that, in this case, the only nonzero $\lambda_m$ is $m=j_1$ or $ m=-j_1$, and we have $(a_m,b_m)=(1,0)$ if $m=j_1$ and $(a_m,b_m)=(0,1)$ if $m=-j_1$, respectively, so $\mathfrak{F}(\rho_{\theta};\partial_{\theta}\rho_{\theta})=0$.
As expected, noise obstructs us from using GHZ-like states for optimal sensing.

We propose a simple local measurement for the ideal probe state $|\psi\ra =\frac{1}{\sqrt{2}}(|J,M\ra + |J,-M\ra)$ which saturates the quantum Cram\'er-Rao bound. Specifically, we consider the observable $\mathbf{D} := \bigotimes_{j=1}^{N} D_j$, where $D_j$ satisfies $D_j \hat{q}^z_j D_j = - \hat{q}^z_j$. 
This is a local measurement since one can measure $\mathbf{D}$ by measuring $D_j$ locally and multiplying the outcomes.
For example, in the basis $\hat{q}^z_j|\Gamma\ra = \Gamma|\Gamma\ra$, where $\Gamma = -s, \cdots s$, one can pick $D_j =\sum_{\Gamma=-s}^s |s\ra\la-s|$, which would give $\mathbf{D}^2 = I$.
The unknown parameter $\theta$ can therefore be obtained by measuring $\la \mathbf{D} \ra = \la \psi_{\theta} |\mathbf{D}|\psi_{\theta} \ra = \cos(2M\theta)$, and we have $\theta = \arccos(\la \mathbf{D} \ra) /(2M)$. 
The variance is $\Gamma_{\mathbf{D}}^2 = \la \mathbf{D}^2\ra - \la \mathbf{D}\ra^2 = 1- \cos^2(2M\theta)$, which gives the standard error of the measurement as $\Delta \mathbf{D} = \Gamma_{\mathbf{D}}/\sqrt{\nu}$, where $\nu$ is the number of samples. 
So we have $\Delta \theta = \frac{\Delta \mathbf{D}}{2M}\frac{\arccos(x)}{dx}|_{x=\la \mathbf{D}\ra} = \frac{1}{2M\sqrt{\nu}}$, which indeed saturates the quantum Cram\'er-Rao bound since QFI $\mathfrak{F}(\psi_{\theta})=4M^2$.

For the partially-erased probe state $\rho_{\theta}=\sum_{m=-j_1}^{j_1} \lambda_m |m\ra \la m|$, recall that the $\theta$ information is encoded in the coherence between the states $|j_2,M\!-\!m\ra$ and $|j_2,-M\!-\!m\ra$. 
A global measurement such as $\mathbf{D}'=\sum_{m=-j_1}^{j_1}|j_2,M\!-\!m\ra \la j_2,-M\!-\!m| + \text{h.c.}$ would be a candidate to estimate the unknown $\theta$.
Indeed, we see $\la \mathbf{D}' \ra = \tr[\rho_{\theta}\mathbf{D}']=\sum_{m=-j_1}^{j_1}a_mb_m \cos(2M\theta)= A \cos(2M\theta)$, where $A:= \sum_{m=-j_1}^{j_1}a_mb_m$. 
Since $\mathbf{D}'^2 = \sum_{m=-j_1}^{j_1} I_m$, we have $\Gamma_{\mathbf{D}'}^2 = 1-A^2 \cos^2(2M\theta)$, and 
\begin{equation}\label{appeq:Delta theta scaling}
    \Delta \theta = \frac{\Gamma_{\mathbf{D}'}}{\sqrt{\nu}2MA \sin(2M\theta)} \sim \frac{1}{2MA}~.
\end{equation}
Note that 
\begin{equation}
    A = \sum_{m=-j_1}^{j_1}a_mb_m=\left[\binom{2J}{J+M}\right]^{-1}\sum_{m_1=-j_1}^{j_1}\binom{2j_1}{j_1+m_1}\left[ \binom{2J-2j_1}{J-j_1+M-m_1} \binom{2J-2j_1}{J-j_1-M-m_1}\right]^{\frac{1}{2}}~,
\end{equation}
and the asymptotic expansion for $A$ can be carried out like in Sec.~\ref{app:proof lemma 3}. We expect $A = 1- O(j_1M^2/J^2)$, which gives us $\Delta \theta \sim M^{-1} = N^{-b}$. 
This scaling matches with the QFI scaling in the red-shaded parameter regime of Fig.~\ref{fig:parameter}(b) in the main text.   

A global measurement is generally hard to implement so it is usually more desirable to use a local measurement.
Below, we show that the local measurement $\mathbf{D}_{\bar{d}} := \bigotimes_{j=d+1}^{N} D_j$ can be used to obtain $\theta$ with a $\Delta \theta$ scaling beyond the standard quantum limit, although it does not saturate the scaling given by the QFI.
Note that, for $\rho_{\theta}=\sum_{m=-j_1}^{j_1}\lambda_m |m\ra \la m|$, if $j_1$ is a half-integer, then we can throw out one more qudit so that $j_1$ is an integer. 
We thus see that $\mathbf{D}_{\bar{d}}$ will pick up the signal from the coherence between states $|j_2\!,\!M\!-\!m\ra$ and $|j_2\!,\!-M-m\ra$ only when $m=0$.
More explicitly, we have $\la\mathbf{D}_{\bar{d}} \ra = \tr[\rho_{\theta}\mathbf{D}_{\bar{d}}] = a_0b_0 \cos(2M\theta)$. 
Comparing this result with Eq.~(\ref{appeq:Delta theta scaling}) and noting that $a_0b_0 \sim (j_1)^{-\frac{1}{2}}$,
we obtain $\Delta \theta \sim \sqrt{j_1}/M \sim N^{-(b-c/2)}$ if $j_1 \sim N^c$ and $M \sim N^b$, exhibiting a scaling beyond the standard quantum limit if $2b-c>1$, together with the conditions $b > c$ and $c < 2(1-b)$. For $c=0$, the local measurement would give a variance scaling matching the QFI scaling.

\section{A Recovery Channel for AQECC}\label{app:recovery}
While a recovery operation is not required for metrological applications, it is nevertheless insightful to understand how a recovery channel can be constructed for our AQECC in the context of quantum error correction. 
Refs.~\cite{benyGeneral2010,benyApproximate2011} provide a general framework for constructing nearly optimal recovery channels, which we briefly outline below.

First, we calculate $\lambda_{ij}$ in the Knill-Laflamme condition. While in our earlier analysis, we used $\lambda_{ij}=\la J,\Mmin|K_i^{\dagger} {K_j}|J,\Mmin\ra$, in fact, any $M$ satisfying $\Mmin \leq M \leq \Mmax$ yields an equally good choice for $\lambda_{ij}=\la J,M|K_i^{\dagger} {K_j}|J,M\ra$ for the purpose of constructing a recovery channel. 
One then considers the following minimization problem
\begin{equation}
    F_{\rho} = \min_{\rho}\tr\sqrt{\Phi_{\rho}(\mathbf{I})}~,
\end{equation}
where
\begin{equation}
    \Phi_{\rho}(\mathbf{I}) = \sum_{ij}K_i\rho^2 K_j^{\dagger}\lambda_{ij}~,
\end{equation}
and the minimization is over the density matrix $\rho = \sum_{a,b=1}^{2^k} \rho_{ab}|J\!,\!M_a\ra\la J\!,\!M_b|$ in the logical space. 

Suppose $\rho_0$ is the density matrix which minimizes the above equation. 
If $\rho_0$ is unique and of full rank, then the following recovery channel is nearly optimal:
\begin{equation}
    \mathcal{R}(\tau)=\Phi_{\rho_0}^{\dagger}[\Phi_{\rho_0}(\mathbf{I})^{-\frac{1}{2}}\tau\Phi_{\rho_0}(\mathbf{I})^{-\frac{1}{2}}] + \mathcal{T}(\tau)~,
\end{equation}
where 
\begin{equation}
    \Phi_{\rho}^{\dagger}(\bullet)=\sum_{ij}\lambda_{ij}[\rho E_j^{\dagger}(\bullet)E_i\rho]~,
\end{equation}
and $ \mathcal{T}(\tau)$ is some completely positive map chosen to ensure that $\mathcal{R}(\tau)$ is trace-preserving.

\section{Spin-1 XY-Dzyaloshinskii-Moriya Scar States as AQECC}\label{app:scars}
In this appendix, we show that the quantum many-body scar states in the spin-1 XY-Dzyaloshinskii-Moriya model~\cite{schecterWeak2019,markUnified,dooleyRobust2021} can form an AQECC. 
While the original model considers a system with $N$ spin-$1$ degrees of freedom, here we generalize it to $N$ spin-$s$ degrees of freedom.
Specifically, we consider the spin-$s$ local Hilbert space with the basis $\{|m\ra, m=-s,\cdots, s\}$. 
We define the operators 
\begin{align}
    J^z=\frac{1}{2}\sum_{j=1}^{N}(|s\ra\la s| - |-s\ra\la -s|)_j~, ~~& J^+ = \sum_{j=1}^{N}e^{i\phi_j}|s\ra\la -s|_j~,~~~~ J^- = \sum_{j=1}^{N}e^{-i\phi_j}|-s\ra\la s|_j~,
\end{align}
where $\phi_j$ is some site-dependent phase.
It is easy to verify that these operators generate the $SU(2)$ algebra $[J^+,J^-]=2J^z$ and $[J^z, J^\pm]=\pm J^\pm$.
The quantum many-body scar states are $|S_M\ra := Z_M(J^{+})^M |S_0\ra$, where $|S_0\ra = \bigotimes_{j=1}^{N}|-\!s\ra_j$ and $Z_M$ is the normalization constant.
The scar states in the spin-1 XY-Dzyaloshinskii-Moriya model correspond to $s=1$ with some specific choice of $\phi_j$.

In fact, these scar states are Dicke states in disguise. This can be seen by redefining $|s\ra_j \rightarrow e^{-\phi_j/2}|s\ra_j$ and $|-s\ra_j \rightarrow e^{\phi_j/2}|s\ra_j$, which can be achieved by applying a tensor product of single-site unitaries $U=\bigotimes_{j=1}^N u_j$, where $u_j=|s\ra \la s|_je^{-\phi_j/2}+|\!-\!s\ra \la\!-\!s|_je^{\phi_j/2} + \sum_{m=-s+1 \cdots s-1}|m\ra \la m|_j$.
We therefore see that $|S_M\ra$ are Dicke states $|J\!,\!M\ra$ ($s=\frac{1}{2}$, $J=\frac{N}{2}$) embedded in a larger local Hilbert space, up to a tensor product of single-site unitary.
Therefore, we conclude that the quantum many-body scar states $\mathfrak{C}=\text{span} \{ |S_M\ra; M=\Mmin, \Mmin+\Delta, \cdots, \Mmax \}$ form an AQECC against $d$-local noise or erasures at known locations if we take $\Delta \geq 2sd+1$. 
The code parameters and the inaccuracy bound against different errors can be readily obtained from Theorems~\ref{theorem 1} and \ref{theorem 2} with $s=\frac{1}{2}$, $J=\frac{N}{2}$.